\documentclass[journal]{IEEEtran}
\ifCLASSINFOpdf
\else
\fi

\usepackage{amsmath}% need for sub equations
\usepackage{amssymb}
\usepackage{cite}
\usepackage{color}   
\usepackage{graphicx}   % need for figures
\usepackage{epsfig} 
\usepackage{mathtools}
\usepackage{marvosym}
\usepackage{soul}
\usepackage{xcolor}
\usepackage{enumitem}

\usepackage{amsthm}
\usepackage{ams fonts} % need for calligraphy fonts (e.g. Range, Null space)
\usepackage{xcolor}
\usepackage{epstopdf}
\usepackage[normalem]{ulem}
\usepackage{cancel}
\usepackage{url}
\usepackage{algorithm}
\usepackage{algpseudocode}
\usepackage{pifont}

\newcommand{\norm}[1]{\left\lVert#1\right\rVert}

\newtheorem{lem}{\bf{Lemma}}
\newtheorem{Def}{\bf{Definition}}
\newtheorem{prop}{\bf{Proposition}}
\newtheorem{thm}{\bf{Theorem}}

\begin{document}
\title{Secure Estimation based Kalman Filter for Cyber-Physical Systems against Adversarial Attacks}

%\author{\IEEEauthorblockN{Young Hwan Chang$^{*}$,}
%\IEEEauthorblockA{Dept. of Electrical Engineering and\\ Computer Sciences,\\Univ. of California, Berkeley, \\CA 94720 USA\\
%Email: yhchang@berkeley.edu}
%\and
%\IEEEauthorblockN{Qie Hu$^{*}$,}
%\IEEEauthorblockA{Dept. of Electrical Engineering and\\ Computer Sciences,\\Univ. of California, Berkeley, \\CA 94720 USA\\
%Email: qiehu@eecs.berkeley.edu}
%\and
%\IEEEauthorblockN{Claire J. Tomlin}
%\IEEEauthorblockA{Dept. of Electrical Engineering and\\ Computer Sciences,\\Univ. of California, Berkeley, \\CA 94720 USA\\
%Telephone: (510) 643--6610\\
%Fax: (510) 643-2356 \\
%Email: tomlin@eecs.berkeley.edu}}
%\footnote{$*$ These authors contributed equally.}


\author{Young~Hwan~Chang$^*$, Qie~Hu$^*$, ~Claire~J.~Tomlin
\thanks{Y. H. Chang is with the Department of Biomedical Engineering, Oregon Health and Science University, Portland, OR 97201 USA (e-mail:$chanyo@ohsu.edu$).}
\thanks{Q. Hu, C.J Tomlin are with the Department of Electrical Engineering and Computer Sciences, University of California, Berkeley, CA 94720 USA (e-mail:$\{qiehu, tomlin\}@eecs.berkeley.edu$).}
\thanks{*These authors contributed equally}}
%
%\author{\IEEEauthorblockN{Young Hwan Chang\IEEEauthorrefmark{1},
%Qie Hu\IEEEauthorrefmark{1}, and
%Claire J. Tomlin}
%%\IEEEauthorblockA{\IEEEauthorrefmark{2}Department of Electrical Engineering and Computer Sciences,\\University of California, Berkeley, CA 94720 USA\\ Email: $\{yhchang@, 
%\IEEEauthorblockA{Department of Electrical Engineering and Computer Sciences,\\University of California, Berkeley, CA 94720 USA\\ Email: $\{yhchang@, 
%qiehu@eecs., tomlin@eecs.\} berkeley.edu$.}
%\IEEEauthorblockA{\IEEEauthorrefmark{1}These authors contributed equally}}


\maketitle

\begin{abstract}
Cyber-physical systems are found in many applications such as power networks, manufacturing processes, and air and ground transportation systems. Maintaining security of these systems under cyber attacks is an important and challenging task, since these attacks can be erratic and thus difficult to model. Secure estimation problems study how to estimate the true system states when  measurements are corrupted and/or control inputs are compromised by attackers. The authors in \cite{Fawzi2014} proposed a secure estimation method when the set of attacked nodes (sensors, controllers) is fixed. In this paper, we extend these results to scenarios in which the set of attacked nodes can change over time. We formulate this secure estimation problem into the classical error correction problem \cite{Candes_Tao} and we show that accurate decoding can be guaranteed under a certain condition.
%show that under a certain condition, this secure estimation problem with time varying attacked nodes is equivalent to the classical error correction problem, and we provide a novel method to guarantee accurate decoding. 
Furthermore, we propose a combined secure estimation method with our proposed secure estimator and the Kalman Filter for improved practical performance. Finally,  we demonstrate the performance of our method through simulations of two scenarios where an unmanned aerial vehicle is under adversarial attack.

\end{abstract}


\begin{IEEEkeywords}
Cyber-physical systems, Error correction, Secure estimation
\end{IEEEkeywords}

% \fi
\IEEEpeerreviewmaketitle

% !TEX root = IEEE_adversarial_attacks.tex
\section{Introduction}

Cyber-physical systems (CPS) %are found in many applications, such as sensor networks, power grids, chemical process plants and Unmanned Aerial Vehicles (UAVs). These systems 
consist of physical components such as actuators, sensors and controllers that communicate with each other over a network. 
For example, unmanned aerial vehicles (UAV) may obtain position measurements from a Global Positioning System (GPS) or communicate with remote control centers.
%For example, we consider a scenario in which a group of unmanned aerial vehicles (UAV) is flying in a formation, and each UAV continuously sends information such as its position to other vehicles wirelessly. Or we can consider each UAV continuously sends information to a remote control center. 
%For example, a UAV typically has onboard controllers and motors that control its rotor speeds and an inertial measurement unit (IMU) that measures position, velocity and acceleration. Often to improve accuracy of position measurements, UAVs also receive Global Positioning Systems (GPS) measurements from satellites \cite{SastryUAV}. 
Although communication networks are often protected by security measures, cyber attacks can still take place when a malicious attacker obtains unauthorized access, launching jamming attacks \cite{Gligor} or spoofing sensor readings and sending erroneous control signals to actuators \cite{Mo}. For CPS, cyber attacks not only compromise information but can also cause damage in the physical process. This presents new challenges and thus demands new strategies and algorithms \cite{Sastry}. 

There has been extensive work on the security of CPS. 
Each of them relies on specific assumptions about attackers' strategies and it is rarely the case, if not impossible, that one estimator/detector can protect against all possible attacks. 
\cite{Tong, KwonACC, Reiter, Sastry2} studied optimal attack strategies for different control systems and applications. 
From the controller's point of view, researchers have studied how to detect attacks \cite{Blanke, Willsky} and how to accurately estimate the states and control the system when it is under attack. One approach for the latter, which is adopted in robust control and filtering methods, is to model the attack signal as process or measurement noise, and assume that they are bounded \cite{Zhou_Doyle} or follow a certain probabilistic distribution \cite{Bullo, Liu}.
%For example, $H_2$ and $H_\infty$ controllers are designed so that the system continues to function properly when it is subjected to bounded model uncertainty or disturbance. Similarly, stochastic controllers \cite{KwonCDC} and filtering methods such as residual filters \cite{Bullo} and the Kalman Filter (KF)  \cite{Liu, Mo} model the attack signal as noise or disturbance that follows a certain probabilistic distribution. %Therefore these controllers may fail to detect attacks that are poorly modeled by a noise process, such as adversarial attacks. 
%As a result, researchers formulated algorithms that treat attack signals more explicitly. 
An alternative approach uses game theory, where the controller and attacker are players with competing goals in a game \cite{Wu, Basar, Basar2, Walrand, Pappas}. %, where the attacker is trying to maximize some cost of loss and the controller is trying to minimize it . %This formulation requires assumptions about the attacker's possible strategies. 
Finally, the authors of \cite{KwonCDC} proposed a hybrid controller, where each constituent controller protects against a specific type of attack.
%Nevertheless, attacks are usually erratic and the assumption that they can be described by a specific model or that the controller is aware of such a model may be unjustified.  %\revision{\sout{However it is not possible to design such a hybrid controller that would protect against all possible attacks.}}

Recently, \cite{Fawzi2014} studied secure estimation of a linear time invariant system where attack signals can be arbitrary and unbounded, thus protecting the system against more general cyber-attacks. Later, \cite{Pajic2014} and \cite{shoukry2016smt} extended this work by relaxing the assumption of having an exact system model and proposing an SMT-based observer that handles large systems with thousands of sensors, respectively. 
One limiting assumption of \cite{Fawzi2014,Pajic2014,shoukry2016smt} is that the set of attacked nodes/sensors is fixed and can not change over time. 
If a malicious attacker is aware of this, then he or she can exploit this weakness and attack different sensors at different time steps so that such a decoder would fail.

In this paper, we are interested in the case in which the set of attacked nodes can change over time. 
By doing so, our proposed decoder can protect the system against more general attack scenarios than that presented in \cite{Fawzi2014}. 
We believe this is a significant contribution, because it is difficult, if not impossible, to anticipate cyber attackers' strategies and behavior, thus a decoder that is able to handle more general attacks is an improvement. 
Furthermore, security studies on the current traffic infrastructure \cite{ghena2014traffic} demonstrated that once a cyber attacker gains access to the traffic network at a single point, the attacker can send commands to any traffic intersection in the network. In other words, the attacker can freely attack a different set of traffic signals (sensors) at any time. Indeed, an attacker who desires to travel through a set of roads as fast as possible would attack different traffic lights to always give him/herself green lights as he/she moves through the road network. 

%We focus on secure estimation and control of linear time invariant systems under sensor attack, because this type of attack is relatively easy to perform and thus particularly interesting. Take UAVs for example: actuators are installed onboard with hardwired local feedback loops, which are unlikely to be corrupted by adversarial attacks.  
%On the other hand, communications with external sources, such as GPS measurements, are much more vulnerable. 

\subsection{Contributions}

There are four main contributions in this paper:
\begin{enumerate}[listparindent=1.5em]
\item
We propose a secure decoder for a linear time invariant system under sensor attack, where the attacked sensors can change with time, and attack signals can be unbounded and arbitrary.
The proposed decoder is based on $l_1$ optimization and is computationally efficient.
\item
We prove the maximum number of sensor attacks that can be corrected with our decoder, which turns out to be the same as that of the decoder proposed in \cite{Fawzi2014}. This is a very nice result as it shows that, compared to \cite{Fawzi2014}, our decoder can protect the system against more general attacks, and at the same time, it does not compromise the number of attacks that can be corrected.
\item
We propose a practical method for decoder design that guarantees accurate decoding. 
First, we formulate the secure estimation problem into the classical error correction (EC) problem \cite{Candes_Tao}. In EC, accurate decoding can be guaranteed if the coding matrix satisfies the Restricted Isometric Properties (RIP), which unfortunately are very difficult to check, in addition, we cannot choose a random coding matrix \textit{a priori} in our problem setting. 
Instead of using RIP, in Theorem 1, we provide a sufficient condition for perfect recovery of the system states against sensor attacks.

%\qie {\yh{TOO LONG AND REPEATED in the main text:} There are two main steps in achieving this:
%We show that under a certain condition, the secure estimation problem for when attacked nodes change with time is equivalent to the classical error correction (EC) problem \cite{Candes_Tao}. In EC, accurate decoding can be guaranteed if the coding matrix satisfies the Restricted Isometric Properties (RIP), which unfortunately, are very difficult to check. In \cite{Candes_Tao}, Candes and Tao proved a beautiful result: by choosing a random coding matrix \textit{a priori}, accurate decoding can be achieved with high probability. 
%However, in our problem, the coding matrix consists of system dynamic matrices and are not arbitrarily. 
%Thus, we describe a pole placement method to guarantee the existence of an accurate decoder in our problem. 
%
%In Theorem 1, we provide a sufficient condition on number of observations needed by our decoder to achieve the maximum number of correctable errors. In addition, we show through extensive simulation, that much fewer observations are required for perfect recovery of the system states. }
%Although this condition may not always be satisfied in practice, simulation results show that our proposed estimator would still recover the correct states with high probability.

\item
Finally, we propose to combine our decoder with a Kalman Filter (KF) to improve its practical performance. 
The KF filters out both occasional estimation errors by the secure decoder and noisy measurements.
We demonstrate the effectiveness of our combined estimator using two examples of UAVs under adversarial attack.
\end{enumerate}

\subsection{Organization of  the Paper}
This paper is organized as follows. Section \ref{sec:overview} gives an overview of secure estimation for CPS when attacked nodes are fixed, as well as compressive sensing and error correction. 
Section \ref{sec:main} formulates the problem of secure estimation when attacked nodes can change over time and compares it with the case when attacked nodes are fixed. 
In Section \ref{sec:design}, we describe our decoder design method %describes explains our proposed secure estimation algorithm 
and assess the decoder's practical performance through extensive simulations. In addition, we describe how to combine it with a KF to improve its performance in practice. Finally we present two more realistic numerical examples of UAVs subject to adversarial attack in Section \ref{sec:examples}. 
%In this last section, we also characterize how pole placement can be used to trade-off control and secure estimation, and show that sensor fusion can improve practical performance. 
In this paper, we focus on estimation of states under sensor attack, hence the terms `estimator' and `decoder', `sensor' and `node', `attack vector' and `error vector' are used interchangeably.
 
\subsection{Notation}
\begin{itemize}
\item 
$\lvert \textsf{supp} (x) \rvert $  denotes the support of vector $x$, i.e., the number of nonzero components in $x$\footnote{If $f$ is any real-valued or vector-valued function on a topological space $X$, the support of $f$, denoted by $\textsf{supp}(f)$, is the closure of the set points where $f$ is nonzero: $\textsf{supp}(f)  = \{ x \in X \,|\, f(x) \neq 0 \}$.}. 
%\item $\norm{M} _{l_0} \triangleq \lvert \textsf{rowsupp} (M) \rvert = \lvert \{ i \in \{ 1, ... , p\} | M_i \neq 0 \}  \rvert 
%	$ where $M$ represents a matrix (not a vector) and $M_i$ represents the $i$-th row of $M$.
\item 
$\norm{ x}_{l_1}:= \sum_{i=1} ^n \lvert x_i \rvert $ where $x \in \mathbb{R}^n$. Note that this is not the same as $\lvert \textsf{rowsupp} ( \cdot )  \rvert $ defined in \cite{Fawzi2014}. 
\item  
For a matrix $M \in \mathbb{R}^{m \times n}$, $\mathcal{N} (M) = \{ x \in \mathbb{R}^n \lvert M x = 0 \}$ represents the null space of $M$. $\mathcal{R}(M)$ denotes the range space of $M$, and is defined as the set of all possible linear combinations of its column vectors.
%\item
%For a square matrix $A \in \mathbb{R}^{n \times n}$, $\lambda_i$ represents its $i^{\text{th}}$ eigenvalue and $v_i \in \mathbb{R}^n$ represents its corresponding eigenvector, i.e., $A v_i = \lambda_i v_i$. In addition, ${\bf v} \triangleq \begin{bmatrix} v_1 & v_2 & \cdots & v_n \end{bmatrix} \in \mathbb{R}^{n \times n}$, $\Lambda \triangleq \text{diag}\{ \lambda_1, \cdots , \lambda_n \}$.
%\item
%If $A \in \mathbb{R}^{n\times n}$ is diagonalizable, then $rank({\bf v}) = n$. Thus the eigenvectors of $A$ form a basis for $\mathbb{R}^n$, and any $z \in \mathbb{R}^n$ can be expressed in the eigenbasis of $A$, i.e., $ z = \sum_{i=1}^n \alpha_i v_i = {\bf v} \cdot \alpha$, where 
%$\alpha = \begin{bmatrix} \alpha_1, \alpha_2, \cdots, \alpha_n \end{bmatrix}^\top $.
%\item 
%$\psi_i \triangleq \begin{bmatrix} Cv_i ; \lambda_i C v_i; \lambda_i^2 C v_i; \cdots; \lambda_i^{T-1} C v_i 
%\end{bmatrix}$ 
%\textcolor{red}{Looking at these again, the last 3 items (starting from 'In addition, ...'), we only use them in the Appendix. What if we moved these last few to the Appendix (after the assumptions)? It may make more sense there?}
\end{itemize}

%\textcolor{red}{Also some of our equations are not "Equations" but "optimization problems", so when we refer to these, do we say "Equation (XX)" or just "(XX)"? The other paper that Claire mentioned all equations to be referred to by "Equation (XX)" didn't have this case.}

%
%\revision{
%\subsection*{Notation:}
%\begin{itemize}
%\item $\lvert \textsf{supp} (x) \rvert $  denotes the support of $x$, i.e., the number of the nonzero components in $x$. 
%\item $\norm{ x}_{l_1}:= \sum_{i=1} ^n \lvert x_i \rvert $ where $x \in \mathbb{R}^n$. Note that this is not $\lvert \textsf{rowsupp} ( \cdot )  \rvert $ shown in \cite{Fawzi2014}. 
%\item
%${\bf v} = \begin{bmatrix} v_1 & v_2 & \cdots & v_n \end{bmatrix}$
%\item
%Since $A$ is diagonalizable, $rank({\bf v}) = n$, thus the eigenvectors of $A$ form a basis for $\mathbb{R}^n$. Consider the decomposition of $z \in \mathbb{R}^n$ in the eigenbasis of $A$, i.e., $ z = \sum_{i=1}^n \alpha_i v_i = {\bf v} \cdot \alpha$ where 
%$\alpha = \begin{bmatrix} \alpha_1, \alpha_2, \cdots, \alpha_n \end{bmatrix}^\top $
%\item
%$\Lambda = \text{diag}\{ \lambda_1, \cdots , \lambda_n \}$
%\item $\psi_i \triangleq \begin{bmatrix} Cv_i ; \lambda_i C v_i; \lambda_i^2 C v_i; \cdots; \lambda_i^{T-1} C v_i \end{bmatrix}$ 
%\end{itemize}
%}

% !TEX root = IEEE_adversarial_attacks.tex

\section{Overview}\label{sec:overview}
%%%%%%%%%%%%%%%%%%%%%%%%%%%%%%%%%%%%%%%%%%
\subsection{Secure Estimation for Fixed Attacked Nodes \cite{Fawzi2014} }% Cyber-Physical Systems under Adversarial Attacks \qie{[This can sound like the formulation here is the definition of secure estimation, but in fact it's only for the case with fixed attacked nodes. Maybe s/t like this:] Secure Estimation for Fixed Attack Nodes} \cite{Fawzi2014}}
Consider a linear dynamical system in the presence of attacks:
\begin{eqnarray}
		x^{(t+1)} &=& A x^{(t)}  \nonumber \\
		y^{(t)} &=& C x ^{(t)} + e^{(t)}
		\label{eq:system_model}
\end{eqnarray}
where $x^{(t)}\in \mathbb{R}^n$ represents the state of the system at time $t\in \mathbb{N}$, $y^{(t)} \in \mathbb{R}^p$ is the output of the sensors at time $t$ and $e^{(t)} \in \mathbb{R}^p $ represents attack signals injected by malicious agents at the sensors. 

In \cite{Fawzi2014}, the authors proposed an elegant state estimation algorithm against adversarial attacks, where the attack signals can be unbounded and arbitrary, but they assumed that the set of attacked nodes $K$ does not change over time. More precisely, if $K \subset \{1, ..., p\}$ is the set of nodes that were attacked, then we have for all $t$, $\textsf{supp}(e^{(t)})\subset K$. % \footnote{if $f$ is any real-valued or vector-valued function on a topological space $X$, the support of $f$, denoted by $\textsf{supp}(f)$, is the closure of the set points where $f$ is nonzero: $\textsf{supp}(f)  = \{ x \in X \,|\, f(x) \neq 0 \}$.}
The problem of reconstructing the initial state $x^{(0)}$ of the plant from the corrupted observations $( y^{(t)})$ was formulated as follows:

\begin{Def} \emph{(\hspace{1sp}\cite{Fawzi2014})} %[see {\cite[Definition 1]{Fawzi2014}}] 
$q$ errors are correctable after $T$ steps by the decoder $\mathcal{D}: {(\mathbb{R} ^p) } ^T  \rightarrow \mathbb{R}^n$ if for any $x^{(0)} \in \mathbb{R}^n$, any $K \subset \{1,.., p\} $ with $\lvert K \rvert \le q$, and any sequence of vectors $e^{(0)},...,e^{(T-1)}$ in $\mathbb{R}^p$ such that $\textsf{supp}(e^{(t)}) \subset K$, we have $\mathcal{D} (y^{(0)},...,y^{(T-1)}) = x^{(0)}$ where $y^{(t)} = CA^t x^{(0)} + e^{(t)}$ for $t=0,...,T-1$.
\end{Def}

\begin{prop}  \emph{(\hspace{1sp}\cite{Fawzi2014})}  \label{prop:Fawzi} %[see {\cite[Proposition 2]{Fawzi2014}}] 
Let $T \in \mathbb{N}  \backslash \{ 0\}$. The following are equivalent:\\
(i) There is a decoder that can correct $q$ errors after $T$ steps;\\
(ii) For all $z\in \mathbb{R}^n \backslash \{0\}$, $\lvert \textsf{supp}(Cz) \cup \textsf{supp}(CAz) \cup \cdots \cup \textsf{supp}(CA^{T-1} z) \rvert > 2q$.
\end{prop}

\noindent The authors then proposed the optimal decoder:
\begin{equation}
x^{(0)} = \text{arg} \min_{x} \norm { Y^{(T)} - \Phi ^{(T)} (x) }_{l_0} 
\label{eq:opt_decoder}
\end{equation}
%\begin{equation}
%\mathcal{D}_0 \bigg (y^{(0)}, y^{(1)}, \cdots y^{(T-1)} \bigg ) = \text{arg} \min_{\hat x \in \mathbb{R}^n} \norm { Y^{(T)} - \Phi ^{(T)} \hat x }_{l_0} 
%\label{eq:opt_decoder}
%\end{equation}
where the $l_0$ ``norm" of matrix $M$ is the number of nonzero rows in $M$ \cite{Fawzi2014}:
\begin{equation}
	\norm{M} _{l_0} \triangleq \lvert \textsf{rowsupp} (M) \rvert = \lvert \{ i \in \{ 1, ... , p\} | M_i \neq 0 \}  \rvert 
	\label{eq:decoder_row}\nonumber
\end{equation}
where $M_i$ represents the $i$-th row of $M$. $Y^{(T)} = \begin{bmatrix} y^{(0)} & \lvert & y^{(1)} & \lvert  & ... & \lvert & y^{(T-1)} \end{bmatrix} \in \mathbb{R}^{p\times T}$ and $\Phi^{(T)}$ is a linear map such that $\Phi ^{(T)} (x) = \begin{bmatrix} Cx  & \lvert & CAx  & \lvert &  ... &  \lvert & CA^{T-1} x \end{bmatrix} \in \mathbb{R}^{p\times T}$. Therefore
\begin{equation}\nonumber
Y^{(T)} - \Phi ^{(T)} (x) =  \begin{bmatrix} e^{(0)}  & \lvert & e^{(1)} & \lvert &  ... &  \lvert & e^{(T-1)} \end{bmatrix},
\end{equation}
which we refer to as the error matrix in the sequel, and consists of horizontally stacked attack/error vectors from time $t=0$ to $t=T-1$.
In other words, the decoder finds the $x^{(0)}$ that minimizes the number of nonzero rows in the resulting error matrix. The foundation for this decoder to work is the assumption of fixed attacked nodes.
As an illustration, consider a system with a single attacked node and construct the error matrix
\begin{equation}
	\begin{bmatrix} e^{(0)}  & \lvert & e^{(1)} & \lvert &  ... &  \lvert & e^{(T-1)} \end{bmatrix} = \begin{bmatrix} 0 & 0 & \cdots & 0 \\
					       * & * & \cdots & * \\
					       0 & 0 & \cdots & 0 \\
					       0 & 0 & \cdots & 0 
			\end{bmatrix}, \nonumber \\
\end{equation}
where $*$ denotes a nonzero component (i.e., attack or corruption). Observe that the set of nonzero rows corresponds to the set of attacked nodes, and furthermore, if a small number of nodes are attacked, then necessarily only a small number of rows are not identically zero and the error matrix has a small row support. The decoder in \cite{Fawzi2014} works by leveraging this property of fixed attacked nodes. 
%The cardinality of the row support of $(Y^{(T)} - \Phi ^{(T)} \hat x)$ used in Equation (\ref{eq:opt_decoder}) is based on the assumption of fixed attacked nodes, because in this case, this cardinality is equal to the number of corrupted nodes. The $l_0$ norm in Equation (\ref{eq:opt_decoder}) can be approximated by an $l_1$ norm to give a convex decoder and thus making it computationally feasible.

What happens if the attacked nodes can change over time? To answer this, let's look at another example where the system has a single attacked node again, however the attacked node is cycled through all nodes such that the error matrix is as follows
\begin{equation}
	\begin{bmatrix} e^{(0)}  & \lvert & e^{(1)} & \lvert &  ... &  \lvert & e^{(T-1)} \end{bmatrix} = \begin{bmatrix} * & 0 & \cdots & 0 \\
					       0 & * & \cdots & 0 \\
					       0 & 0 & \cdots & 0 \\
					       0 & 0 & \cdots & * 
			\end{bmatrix}. \nonumber \\
\end{equation}
Observe that although a single node is attacked at any time $t$, since the attacked node is changing over time, the error matrix has full row support. Therefore a decoder that attempts to find the optimal $x^{(0)}$ by minimizing the row support of the error matrix does not work here. One may argue that this decoder can be used with $T=1$, to solve a secure estimation problem where the attacked nodes change over time. 
However, when $T=1$, the number of correctable errors cannot be large.
From Proposition 1 (i.e. Proposition 2 in \cite{Fawzi2014}), ``the decoder can correct $q$ errors after $T=1$ step if and only if for all $z \in \mathbb{R}^n \backslash \{0\}$, $\lvert \textsf{supp} (Cz) \rvert > 2q$''. 
As an example, for any $C \in \mathbb{R}^{p \times n}$ where $p<n$, $C$ has a nontrivial null space, hence there exists $z \in \mathbb{R}^n \backslash \{ 0 \}$ such that $ \lvert \textsf{supp} (Cz) \rvert = 0$; in other words, zero errors are correctable. Therefore, the decoder for fixed attacked nodes cannot be easily extended to when the attacked nodes can change over time, and a new decoder is required for the latter.

In this paper, we propose such a new decoder, one that is suitable for when the set of attacked nodes can change over time. Our decoder is inspired by the error correction problem \cite{Candes_Tao} and by observing that if we construct the error matrix by stacking the error vectors $e^{(0)}, \cdots, e^{(T-1)}$ vertically, instead of horizontally, then even when the attacked nodes can change over time, as long as the number of attacked nodes at each time $t$ is small, the error matrix (which is a large error vector in this case) is sparse. Before we dive into details of this method and its properties, we first give a brief overview about compressive sensing and error correction in the following section.

%%%%%%%%%%%%%%%%%%%%%%%%%%%%%%%%%%%%%%%%%%%%%%%%%%%%%%%%%%%%%%%%%%%%%%%%%%%%%%%%%%%%%%%%%%%%%%%%%%%%%%%%%%%%%%

\subsection{Overview: Compressive Sensing and the Error Correction Problem \cite{Candes_Tao}} %\qie{either `and error correction' or `and the error correction problem'} \cite{Candes_Tao}}

\subsubsection{Compressive Sensing}
Sparse solutions $x\in \mathbb{R}^n$, are sought to the following problem:
\begin{equation}
	\min_x \norm{x}_0 \text{ subject to } b= Ax
	\label{eq:CS}
\end{equation}
where $b \in \mathbb{R}^m$ are the measurements, $A \in \mathbb{R}^{m\times n}~ (m \ll n)$ is a sensing matrix and $\norm{x}_0$ denotes the number of nonzero elements of $x$. The following lemma provides a sufficient condition for a unique solution to (\ref{eq:CS}) \cite{Candes_Tao}:

\begin{lem} \emph{(\hspace{1sp}\cite{David_Chang})} \label{lem:CS}
If the sparsest solution to (\ref{eq:CS}) has $\norm{x}_0 = q$, $m\ge 2q$ and all subsets of $2q$ columns of $A$ are full rank, then the solution is unique. 
\end{lem}
\begin{proof}
Suppose the solution is not unique, hence there exists $x_a \neq  x_b$ such that $Ax _1 = b$ and $Ax_2 = b$ where $\norm{x_1}_0 = \norm{x_2}_0 = q$. Then $A(x_1 - x_2) = 0$ and $x_1 - x_2 \neq 0$. Since $\norm{x_1-x_2}_0 \leq 2q$ and all $2q$ columns of $A$ are full rank (i.e. linearly independent), it is impossible to have $x_1-x_2\neq 0$ that satisfies $A(x_1-x_2) = 0$ (contradiction).
\end{proof}

%%%%%%%%%%%%%%%%%%%%%%%%%%%%%%%%%%%%%
\subsubsection{The Error Correction Problem  \cite{Candes_Tao}} \label{sec:error_correction}
Consider the classical error correction problem: $y=Cx + e$ where $C\in \mathbb{R}^{p\times n}$ is a coding matrix $(p > n)$ and assumed to be full rank. We wish to recover the input vector $x \in \mathbb{R}^n$ from corrupted measurements $y$. Here, $e$ is an arbitrary and unknown sparse error vector. To reconstruct $x$, note that it is obviously sufficient to reconstruct the vector $e$ since knowledge of $Cx + e$ together with $e$ gives $Cx$, and consequently $x$ since $C$ has full rank \cite{Candes_Tao}. In \cite{Candes_Tao}, the authors construct a matrix $F$ which annihilates $C$ on the left, i.e.  $FCx = 0$ for all $x$. Then, they apply $F$ to the output $y$ and obtain
\begin{equation}
	\tilde y = F (Cx + e) = Fe 
\end{equation}
Thus, the decoding problem can be reduced to that of reconstructing a sparse vector $e$ from the observations $\tilde y = Fe$. Therefore, by Lemma \ref{lem:CS}, if all subsets of $2q$ columns of $F$ are full rank, then we can reconstruct any $e$ whose  $\lvert \textsf{supp}(e) \rvert \le q$. We refer to a decoder that can correct $q$ errors as a $q$-error-correcting decoder.

%%%%%%%%%%%%%%%%%%%%%%%%%%%%%%%%%%%%%
\subsubsection{ Equivalence between the two programs $l_0$ and $l_1$}\label{sec:equiv}
In \cite{Candes_Tao}, the authors mentioned that the computational intractability of the $l_0$-program led researchers to develop alternatives, and a frequently discussed approach considers a similar program in the $l_1$ norm which goes by the name of Basis Pursuit. Motivated by the problem of finding sparse decompositions of special signals in the field of mathematical signal processing, a series of beautiful and ground breaking works \cite{Donoho2003, Elad2002, Gribonval2003, Tropp2004} showed exact equivalence between the two programs $l_0$ and $l_1$ when the RIP conditions are satisfied. Therefore, the $l_0$ norm in (\ref{eq:CS}) can be approximated by an $l_1$ norm to give a convex decoder and is therefore computationally feasible. 
We will discuss in more detail the conditions required to ensure accurate decoding using an $l_1$-optimization based decoder in Section \ref{sec:design}.

% !TEX root = IEEE_adversarial_attacks.tex
%%%%%%%%%%%%%%%%%%%%%%%%%%%%%%%%%%%%%%%%%%%%%%%%%%%%%%%
%%%%%%%
\section{Secure Estimation when Attacked Nodes can Change with Time}  \label{sec:main}
%%%%%%%%%%%%%%%%%%%%%%%%%%%%%%%%%%%%%%%%%%%%%%%%%%%%%%%%%%%
%\revision{ (Put Notations somewhere earlier in the paper? Because notations such as  $\lvert \textsf{supp} (x) \rvert$ appear in sections I and II. If so, we can put all notation definitions there, e.g. rowsupp, N(F), R(C), etc.)
%\subsection*{Notation:}
%\begin{itemize}
%\item $\lvert \textsf{supp} (x) \rvert $  denotes the support of $x$, i.e., the number of nonzero components in $x$. 
%\item $\norm{ x}_{l_1}:= \sum_{i=1} ^n \lvert x_i \rvert $ where $x \in \mathbb{R}^n$. Note that this is not $\lvert \textsf{rowsupp} ( \cdot )  \rvert $ shown in \cite{Fawzi2014}. 
%\item
%${\bf v} = \begin{bmatrix} v_1 & v_2 & \cdots & v_n \end{bmatrix}$
%\item
%Since $A \in \mathbb{R}^{n\times n}$ is diagonalizable, $rank({\bf v}) = n$, thus the eigenvectors of $A$ form a basis for $\mathbb{R}^n$. Consider the decomposition of $z \in \mathbb{R}^n$ in the eigenbasis of $A$, i.e., $ z = \sum_{i=1}^n \alpha_i v_i = {\bf v} \cdot \alpha$ where 
%$\alpha = \begin{bmatrix} \alpha_1, \alpha_2, \cdots, \alpha_n \end{bmatrix}^\top $
%\item
%$\Lambda = \text{diag}\{ \lambda_1, \cdots , \lambda_n \}$ where $\lambda_i$ represents the eigenvalues of $A$.
%\item $\psi_i \triangleq \begin{bmatrix} Cv_i ; \lambda_i C v_i; \lambda_i^2 C v_i; \cdots; \lambda_i^{T-1} C v_i \end{bmatrix}$ 
%\end{itemize}
%}

\subsection{Problem Formulation}
Consider the linear time invariant system as follows:
\begin{equation}
\begin{aligned}
x^{(t+1)} &= A_o x^{(t)} + B u^{(t)} \\
y^{(t)} &= C x ^{(t)} + e^{(t) }
\end{aligned} 
\label{eq:system_model_se}
\end{equation} 
where $x^{(t)} \in \mathbb{R}^n$, $y^{(t)}  \in \mathbb{R}^p$ and $u^{(t)} \in \mathbb{R}^m$ are the states, measurements and control inputs at time $t$. $e^{(t)} \in \mathbb{R}^p$ is the attack signal, and we assume that the attacked nodes can change over time. We define the number of correctable errors as follows:

\begin{Def}\label{def:num_err_change}
When the set of attacked nodes can change over time, $q$ errors are correctable after $T$ steps by the decoder $\mathcal{D}: {(\mathbb{R} ^p) } ^T  \rightarrow \mathbb{R}^n$ if for any $x^{(0)} \in \mathbb{R}^n$ and any sequence of vectors $e^{(0)},...,e^{(T-1)}$ in $\mathbb{R}^p$ such that $\lvert \textsf{supp}(e^{(t)}) \rvert \leq q$, 
we have $\mathcal{D} (y^{(0)},...,y^{(T-1)}) = x^{(0)}$ where $y^{(t)} = CA^t x^{(0)} + e^{(t)}$ for $t=0,...,T-1$.
\end{Def}

In addition, assume that a local control loop implements secure state feedback: $u^{(t)} = Gx^{(t)}$, i.e. the local control loop is not subject to attack. The resulting closed loop system matrix is $A$ $(=A_o+BG)$. 
In practice, this represents the following scenario: a physical system possesses a local control loop that has direct access to the state of the plant and can control the evolution of the physical system. This is reasonable if the sensors are connected to the local controller through a wired link that is not subject to external attacks. Also, as part of the overall plant, a higher-level supervisory and monitoring system receives measurements from the sensors through wireless and vulnerable communication links that are subject to attacks \cite{Fawzi2014}. 
A concrete example is a UAV that uses measurements from onboard, hardwired sensors such an Inertial Measurement Unit (IMU) for its autopilot and trajectory following (i.e. secure local control loop), and communicates wirelessly with a remote control center (i.e. vulnerable link subject to attack).

%%%%%%%%%%%%%%%%%%%%%%%%%%%%%%%%%%%%%%%%%%%%%%%%%%%%%%%%
%%%%%%%%%%%%%%%%%%%%%%%%%%%%%%%%%%%%%%%%%%%%%%%%%%%%%%%

%\subsection{Formulation} \label{sec:err_corr}
\subsection{Reformulation into an Error Correction Problem} \label{sec:err_corr}
The problem that we want to solve is how to reconstruct the initial state $x^{(0)}$ of the system from the corrupted observations ($y^{(t)}$) where $t=0,...,T-1$.
Let $E_{q,T}$ denote the large vector with error vectors from time $t=0$ to $t=T-1$ stacked vertically: $\begin{bmatrix} e^{(0)}; ~ ...~  ;  e^{(T-1)} \end{bmatrix}   \in  \mathbb{R}^{p\cdot T} $, where each $e^{(t)}$ satisfies $\lvert \textsf{supp}(e^{(t)}) \rvert \le q \le p $ as in Definition \ref{def:num_err_change}. 
%Note that the set of attacked nodes can change over time.%\revision{. \sout{ and our goal is to find the initial state $x^{(0)}$.}}
\begin{eqnarray} \label{eq:sys_err_corr}
\begin{aligned}
	Y &\triangleq \begin{bmatrix} y^{(0)} \\ y^{(1)} \\ \vdots \\ y^{(T-1)} \end{bmatrix} 
		= \begin{bmatrix} Cx^{(0)} + e^{(0)} \\ CA x^{(0)} + e^{(1)} \\ \vdots \\ CA^{T-1} x^{(0)} + e^{(T-1)} \end{bmatrix} \\
		& =
		\begin{bmatrix} C \\ CA \\ \vdots \\ CA^{T-1} \end{bmatrix} x^{(0)} + E_{q,T} \triangleq \Phi x^{(0)} + E_{q,T}
		\label{eq:decoder_Phi}
\end{aligned}
\end{eqnarray}
where $Y \in \mathbb{R}^{p\cdot T}$ is the set of vertically stacked corrupted measurements and $\Phi \in \mathbb{R}^{p\cdot T \times n}$ represents the observability matrix of the closed loop system. %\st{ if $B=0$ (i.e., $A = A_o$)}. \st{Since we consider secure estimation of adversarial attacks,} 
Assume $\text{rank}(\Phi) = n$, which is reasonable because otherwise, the system is unobservable and thus $x^{(0)}$ cannot be determined even if there is no attack ($E_{q,T} = 0$).
\begin{itemize}
\item Open-loop case $(B=0)$:  A full column rank condition represents the pair $(A_o,C)$ being observable. In other words, if not, one cannot reconstruct $x^{(0)}$ even if there are no errors in the measurements.
\item State-feedback case: Since state-feedback may affect the observability of a system (even though the pair $(A_o,C)$ is observable), we have to satisfy $\text{rank}(\Phi) = n$ for the closed-loop system with state-feedback. %Recall that $A$ represents the closed-loop system matrix $(A_o + BG)$.
\end{itemize}
Note that the closed-loop system with state-feedback is controllable if and only if the open-loop system is controllable. However, state-feedback may affect the observability of a system. 

%\yh{I think it would be good to move Proposition 2 before "Next, we present two method ..." because Proposition 2 is the condition for L0. Also, since we did not show that our condition is the exact same as RIP, I think we cannot claim we satisfy RIP condition.} 
Proposition \ref{prop:equivalent} below, gives two equivalent sufficient conditions for the existence of a unique solution $x^{(0)}$ to (\ref{eq:decoder_Phi}).

\begin{prop} \label{prop:equivalent}
	%Given the RIP conditions are satisfied \qie{Is it correct to include this condition here to ensure equivalence of L0 and L1?} \yh{Probably not. This is for L0 }, 
	Given $Y=\Phi x + E$, $\Phi \in \mathbb{R}^{p\cdot T \times n}$ is full rank and $2s = 2q\cdot T \leq  p\cdot T-n$, then the following are equivalent: \\% for $x ^{(0)} $ to be the unique solution of (\ref{eq:direct_l1}):\\
$(i)$ all subsets of $2s$ columns of $Q_2 ^\top$ are linearly independent, where  $\Phi = \begin{bmatrix} Q_1 & Q_2 \end{bmatrix} \begin{bmatrix} R_1 \\ 0 \end{bmatrix}$ is the QR decomposition of $\Phi$ and $Q_2 ^\top \in \mathbb{R}^{(p\cdot T-n)\times p\cdot T}$;\\
$(ii)$ $\lvert \textsf{supp}( \Phi z) \rvert > 2 s$ for all $z \in \mathbb{R}^n \backslash \{ 0 \}$.
\end{prop}
\begin{proof}
%\yh{ REMOVE: The RIP conditions are satisfied, thus from Section \ref{sec:equiv} the $l_0$ and $l_1$ programs are equivalent. Then by Lemmas \ref{lem:CS} and \ref{lem:equivalent} we have that $x ^{(0)} $ is the unique solution of (\ref{eq:direct_l1}) if condition $(i)$ is satisfied. 
%Next, we prove $(i) \Leftrightarrow (ii)$.}
$(ii) \implies (i)~:$ Suppose there exist $2s$ columns of $Q_2^\top$ that are linearly dependent. Then, there exists $E_0 \neq 0$ such that $Q_2^\top E_0 = 0$ where $\lvert \textsf{supp}(E_0) \rvert \le 2s$. Observe that $\mathcal{N}(Q_2^\top) = \mathcal{R}(\Phi)$, therefore there exists $z$ such that $E_0 = \Phi z$ (i.e., $E_0 \in \mathcal{R}(\Phi)$). Then, $ \lvert  \textsf{supp}(\Phi z) \rvert = \lvert \textsf {supp} (E_0) \rvert \le 2s $ (contradiction).

$(i) \implies (ii)$: We again resort to contradiction. Suppose that there exists $z\neq 0$ such that $\lvert \textsf{supp}(\Phi z)\rvert \le 2s$. Let $L_1$ and $L_2$ be two disjoint subsets of $\{1,..., p\cdot T \}$ with $\lvert L_1 \rvert \le s$ and $\lvert L_2 \rvert \le s$ such that $L_1 \oplus L_2 = \textsf{supp}(\Phi z)$ (such $L_1$ and $L_2$ exist since $\lvert \textsf{supp}(\Phi z) \rvert \le 2s < p\cdot T$). Let $E_1 = \Phi z \lvert _{L_1}$ be the vector obtained from $\Phi z$ by setting all the components outside of $L_1$ to 0, and similarly let $E_2= - \Phi z \lvert _{L_2}$ (i.e., $E_1 \neq E_2$). Then we have $\Phi z = E_1 - E_2$ with $\textsf{supp}(E_1) \subset L_1$ and $\textsf{supp}(E_2) \subset L_2$ with $\lvert L_1 \rvert \le s$ and $\lvert L_2 \rvert \le s$. Now, consider $Y=\Phi z + E$
\begin{equation}
\begin{aligned}
	\tilde Y &= Q_2^\top Y =  Q_2^\top (\Phi z+E) =  Q_2^\top (E_1 - E_2 + E  ) \\
	& =  Q_2^\top(E_1-E_2) +  Q_2^\top E =  Q_2^\top E \\
	& \implies  Q_2^\top (E_1-E_2) = 0 \nonumber 
\end{aligned}
\end{equation}
The last equality is due to $\mathcal{N}(Q_2^\top) = \mathcal{R}(\Phi)$ (i.e. $Q_2^\top \Phi = 0$). Since all subsets of $2s$ columns of $Q_2^\top$ are linearly independent, $E_1 - E_2 = 0 $ (contradiction).
\end{proof}

Next, we present two methods, inspired by error correction techniques \cite{Candes_Tao}\cite{David_Chang}, for estimating $x^{(0)}$. We show their equivalence and provide sufficient conditions for the existence of a unique solution.

\noindent
{\bf Decoder 1:~} The first method determines the error vector $E_{q,T}$ first, and then solves for $x^{(0)}$. %[We say this upfront, so readers know what to expect, and it may be less surprising when they see the second method?]  \st{We consider the error correction approach.} 
Consider the $QR$ decomposition of $\Phi \in \mathbb{R}^{p\cdot T \times n}$,
\begin{eqnarray}
	\Phi = \begin{bmatrix} Q_1 & Q_2 \end{bmatrix} \begin{bmatrix} R_1 \\ 0 \end{bmatrix} = Q_1 R_1 
\end{eqnarray}
where $\begin{bmatrix} Q_1 & Q_2 \end{bmatrix} \in \mathbb{R}^{p\cdot T \times p\cdot T}$ is orthogonal, $Q_1 \in \mathbb{R}^{p\cdot T\times n}, Q_2 \in \mathbb{R}^{p\cdot T \times (p\cdot T-n)}$ and $R_1 \in \mathbb{R}^{n\times n}$ is a rank-$n$ upper triangular matrix. Pre-multiply (\ref{eq:decoder_Phi}) by $\begin{bmatrix} Q_1 & Q_2 \end{bmatrix} ^\top$:
\begin{equation}
	\begin{bmatrix} Q_1 ^\top \\ Q_2 ^\top \end{bmatrix} Y = \begin{bmatrix}R_1 \\ 0  \end{bmatrix} x^{(0)} + \begin{bmatrix} Q_1 ^\top \\ Q_2^\top \end{bmatrix} E_{q,T}.
	\label{eq:QR}
\end{equation}
We can now solve for $E_{q,T}$ using the second block row:
\begin{equation}
	\tilde Y \triangleq Q_2^\top Y = Q_2^\top E_{q,T}
	\label{eq:E_est}
\end{equation}
where $Q_2^\top \in \mathbb {R} ^{ (p\cdot T-n) \times p\cdot T}$. From Lemma \ref{lem:CS}, Equation (\ref{eq:E_est}) has a unique, $s(\le q\cdot T)$-sparse solution if all subsets of $2s(\le2 q\cdot T)$ columns of $Q_2^\top$ are full rank (this is a reasonable assumption if $ (p\cdot T-n) \ge 2s = 2q\cdot T$). 
Therefore, the secure decoder has two steps, we first solve the following $l_1$-minimization problem
\begin{equation}
	\hat E_{q,T} = \arg \min_E \norm { E}_{l_1} \text{ subject to } \tilde Y = Q_2^\top E 
	\label{eq:solve_E}
\end{equation}
Note that the $l_0$ norm in (\ref{eq:CS}) was replaced with with the $l_1$ norm in (\ref{eq:solve_E}) to give a convex program, which was proposed by Candes and Tao in \cite{Candes_Tao}. In addition, in Section \ref{sec:equiv} we discussed the conditions for equivalence of these two programs.

With $\hat E_{q,T}$ in hand, we then solve for $x^{(0)}$ in the second step using the first block row of Equation (\ref{eq:QR}):
\begin{equation}
	x^{(0)} = R_1^{-1} Q_1^\top (Y- \hat E_{q,T})
	\label{eq:QR1}
\end{equation}

\noindent
{\bf Decoder 2:~} 
The second method recovers $x^{(0)}$ from the corrupted data $Y$ directly by %\st{, we consider} 
solving the following $l_1$-minimization problem \cite{Candes_Tao}:
\begin{equation}
	\min_x \norm { Y  - \Phi x}_{l_1}
	\label{eq:direct_l1}
\end{equation}

\begin{lem} \label{lem:equivalent}
 $x^{(0)}$ is the unique solution of (\ref{eq:direct_l1}) if and only if ${\hat E}_{q,T}$ is the unique solution of (\ref{eq:solve_E}).
\end{lem}
\begin{proof} (By \cite{Candes_Tao}) Observe that on one hand, since $Y = \Phi x^{(0)} + E_{q,T}$ and we may decompose $x = x^{(0)} + v$, hence 
\begin{equation}
	(\ref{eq:direct_l1}) \Leftrightarrow \min_v \norm{  E_{q,T} - \Phi v }_{l_1} \nonumber 
\end{equation}
On the other hand, the constraint $Q_2^\top E = \tilde Y = Q_2^\top E_{q,T}$ means that $E = E_{q,T} - \Phi v $ for some $v \in \mathbb{R}^n$ and, therefore,
\begin{eqnarray}
	(\ref{eq:solve_E}) &\Leftrightarrow& \min_v \norm{ E }_{l_1}, ~~~~ E = E_{q,T}-\Phi v  \nonumber \\
				 & \Leftrightarrow& \min_v  \norm{E_{q,T}-\Phi v}_{l_1} \nonumber 
\end{eqnarray}
Thus, (\ref{eq:solve_E}) and (\ref{eq:direct_l1}) are equivalent programs \cite{Candes_Tao}.
\end{proof}
\noindent Even though we are interested in the state $x^{(0)}$ and not necessarily the error vectors $E_{q,T}$, Lemma \ref{lem:equivalent} states that if the attack vectors cannot be uniquely determined from (\ref{eq:solve_E}), then we cannot estimate $x^{(0)}$ uniquely from (\ref{eq:direct_l1}). \cite{Fawzi2014} also mentioned this notion: the existence of a decoder that can correct $q$ errors is equivalent to saying that the map, $\begin{bmatrix} \Phi & {\bf I}_{q\cdot T} \end{bmatrix}: \mathbb{R}^n \times E_{q,T} \rightarrow (\mathbb{R}^p)^T$ has an inverse for the first $n$ components of its domain where $Y = \begin{bmatrix} \Phi & {\bf I}_{q\cdot T} \end{bmatrix} \begin{bmatrix} x^{(0)} \\ E_{q,T} \end{bmatrix} $ since the attack vectors are uniquely determined by $x^{(0)}$ and the $y^{(t)}$'s, i.e., $e^{(t)} = y^{(t)} - CA ^t x^{(0)}$.

\subsection{Comparison with Secure Estimation for Fixed Attacked Nodes}
We refer to a decoder that can correct $q$ errors as a $q$-error-correcting decoder.
It is interesting to compare the conditions for the existence of a $q$-error-correcting decoder for when the attacked nodes are fixed (Proposition \ref{prop:Fawzi}) and when the attacked nodes can change over time (Proposition \ref{prop:equivalent}) are :
$\forall z \in \mathbb{R}^n\backslash \{0 \}$,
\begin{equation}
\begin{aligned}
(i) &~  \lvert \textsf{supp}(Cz) \cup \textsf{supp}(CAz) \cup \cdots \cup \textsf{supp}(CA^{T-1} z) \rvert > 2q  \\& \text{ (the set of attacked nodes is fixed)}\\
(ii) & ~ \lvert \textsf{supp} (\Phi z) \rvert % = \lvert \textsf{supp} \bigg (\begin{bmatrix} Cz \\ CA z \\ \vdots \\ CA^{T-1} z \end{bmatrix} \bigg) \rvert 
     =  \sum_{i=0}^{T-1} \lvert \textsf{supp} (C A^i z) \rvert > 2 q \cdot T \\& \text{ (the set of attacked nodes can change)} 
	% \\&&\lvert \textsf{supp}(Cz) \cup \textsf{supp}(CAz) \cup \cdots \cup \textsf{supp}(CA^{T-1} z) \rvert  > 2 q \cdot T \nonumber  %\\
%(iii) && \lvert \textsf{supp} (Cz) \rvert > 2q,  \lvert \textsf{supp} (CAz) \rvert > 2q, ..., \text{ and } \lvert \textsf{supp} (CA^{T-1})z) \rvert > 2q \nonumber
	\label{eq:connection}
\end{aligned}
\end{equation}
It is easy to see that when $T=1$, conditions $(i)$ and $(ii)$ are equivalent as both of them reduce to $\lvert \textsf{supp} (Cz) \rvert > 2q$ for all $z \in \mathbb{R}^n\backslash \{0 \}$.
When $T>1$ (i.e., with dynamics), the comparison is not so straightforward. In \cite{Fawzi2014}, the authors proved an equivalent condition for $(i)$:

\begin{lem} \label{lem:distinct}
Assume $A$ has $n$ distinct positive eigenvalues ($0<\lambda_1 < \lambda_2 <\cdots < \lambda_n$) and $T \geq n$. Then, the following are equivalent:
\begin{equation}
\begin{aligned}
 (i) &~ \forall z \in \mathbb{R}^n \backslash \{0\},  \\& \lvert \textsf{supp}(Cz) \cup \textsf{supp}(CAz) \cup \cdots \cup \textsf{supp}(CA^{T-1} z) \rvert > 2q  \\
 (ii)  &~\forall v_i \in \mathbb{R}^n \text{ where }   Av_i = \lambda_i v_i \text{ (i.e., eigenvector of $A$)},\\& ~ \lvert \textsf{supp}(Cv_i) \rvert > 2q \nonumber
\label{eq:condition}
\end{aligned}
\end{equation}
\end{lem}
\begin{proof} Refer to the proof in \cite{Fawzi2014}.
\end{proof}
\noindent
The significance of this lemma is that in order to check whether a decoder can guarantee accurate decoding of $q$ errors when the attacked nodes are fixed, one no longer needs to check satisfiability of condition $(i)$ which is stated for all $z \in \mathbb{R}^n \backslash \{0\}$ and hard to check, instead, one can simply check condition $(ii)$ for the eigenvectors of $A$ which is much simpler. Next, we derive a similar result for our decoder for when the attacked nodes can change with time.

\begin{thm} Let $A \in \mathbb{R}^{n\times n}, C \in \mathbb{R}^{p\times n}$. Assume that $C$ is full rank, $(A,C)$ is observable and $A$ has $n$ distinct positive eigenvalues such that $0 < \lambda_1 < \lambda_2 < \cdots < \lambda_n$. %, \st{and each row of $C$ is not identically zero and is not redundant } . 
Define:
\begin{itemize}
\item
$s_i \triangleq \lvert \textsf{supp} (Cv_i) \vert$, where $v_i$ is an eigenvector of $A$, %($Av_i =\lambda_i v_i$),
\item
$\mathcal{S} \triangleq \{ s_1, s_2, \cdots, s_n \}$,
\item
For every $m \in \{2, \ldots, n\}$, let $\mathcal{S}_m$ be any subset of $\mathcal{S}$ with $m$ elements, define $T_{\mathcal{S}_m} \triangleq \frac {  (m-2) \cdot p + \min \mathcal{S}_m } {\max \mathcal{S}_m - 2q }$.
Then $T_m$ is such that $T_m > T_{\mathcal{S}_m}$ for all subsets $\mathcal{S}_m$, i.e. all subsets of $m$ elements from the set $\mathcal{S}$.
\end{itemize}
Choose $T$ such that  $T \ge \max \{ T_2, \cdots, T_n \}$.
Then, the following are equivalent:
%(for $m=1$, we have shown that (i) and (ii) in Equation (14) are equivalent where $T=1$):
\begin{equation}
\begin{aligned} 
 (i)  &~\forall v_i \in \mathbb{R}^n \text{ where } Av_i =\lambda_i v_i, ~ \lvert \textsf{supp}(Cv_i) \rvert > 2q  \\
  (ii)  &~\forall v_i \in \mathbb{R}^n \text{ where } Av_i =\lambda_i v_i,  ~\lvert \textsf{supp} (\Phi v_i) \rvert > 2q \cdot T  \\
  (iii) &~  \forall z \in \mathbb{R}^n\backslash \{0 \}, \lvert \textsf{supp} (\Phi z) \rvert > 2 q \cdot T \\&~~~~~~~~ \text{(i.e. condition (ii) in Equation (\ref{eq:connection}))} \nonumber 
\label{eq:new_condition}
\end{aligned}
\end{equation}
\end{thm}
\noindent
In order to prove Theorem 1, we make use of Lemmas \ref{lem:two_vec}, \ref{lem:three_vec} and Proposition \ref{prop:m_vec} (see Appendix): 

\begin{proof} (Proof of Theorem 1)

First, it is simple to prove that $(i)$ and $(ii)$ are equivalent: $\lvert \textsf{supp} (\Phi v_i) \rvert = \sum_{k=0}^{T-1} \lvert \textsf{supp} (CA^k v_i ) \rvert = \sum_{k=0}^{T-1} \lvert  \textsf{supp} (\lambda_i^k C v_i )\rvert  = T\cdot  \lvert \textsf{supp} (C v_i) \rvert  $. 

Second, we want to show that $(ii)$ and $(iii)$ are equivalent. The direction $(iii) \implies (ii)$ is trivial, since $(ii)$ is a specific case of (iii) with $z = v_i$. The other direction is more complex. Note that $A$ is diagonalizable, therefore its eigenvectors form a basis for $\mathbb{R}^n$. Now consider the decomposition of $z $ in the eigenbasis of $A$, i.e. $z = \sum_{i=1}^n \alpha_i v_i$ with $\alpha_i \neq 0$ for at least one $i$. 
\begin{enumerate}
\item $m=1$: Suppose there exists $z \in \mathbb{R}^n \backslash \{ 0\}$ such that $ \lvert \textsf{supp} (\Phi z) \rvert \le 2 q \cdot T$. Without loss of generality, let $\alpha_i \neq 0$ and $\alpha_j= 0$ for all $j \neq i$, then, $2 q \cdot T \ge \lvert \textsf{supp} (\Phi z) \rvert = \lvert \textsf{supp} (\alpha_i \Phi v_i) \rvert = \lvert \textsf{supp} (\Phi v_i) \rvert $ for all $T$ (contradiction, $\because \forall v_i, \lvert \textsf{supp} (\Phi v_i) \rvert  > 2 q \cdot T$). 
\item $m=2$: By Lemma \ref{lem:two_vec}, if we choose $T>T_2$.
\item $m\ge 3$: By Lemma \ref{lem:three_vec} and Proposition \ref{prop:m_vec}, if we choose $T>T_m$ for each value of $m$, respectively.
\end{enumerate}
We need to choose $T$ to satisfy the worst case for any $m$ such that $n \ge m \ge 2$. Thus, if $T \ge \max \{ T_2, \cdots, T_n \}$, then $(ii)$ and $(iii)$ are also equivalent.
\end{proof}

%\begin{lem} \label{lem:controllability}
%Assume that the pair $(A_o,B)$ is controllable. Then the closed-loop system with state feedback is controllable and thus, there exists $G $ such that the eigenvalues of the closed-loop matrix $A$ $(=A_o+BG)$, i.e., $\lambda_1, ..., \lambda_n$ can be arbitrarily located on the complex plane.
%\end{lem}

\begin{prop}\label{prop:equivalent2}
Given $C$ is full rank, the closed-loop matrix $A~(=A_o+BG)$ has $n$ distinct positive eigenvalues, the open-loop pair $(A_o,B)$ is controllable, the closed-loop pair $(A,C)$ is observable and $T$ is chosen to satisfy Theorem 1. Then, the condition for secure estimation of $q$-errors when the set of attacked nodes is fixed ((i) in (\ref{eq:connection})) is the same as the condition for when the set of attacked nodes can change over time ((ii) in (\ref{eq:connection})), except the condition on $T$.
\end{prop}
\begin{proof}
Since the pair $(A_o,B)$ is controllable, there exists a feedback matrix $G$ such that the eigenvalues of the closed-loop matrix $A$, i.e., $\lambda_1, ..., \lambda_n$ can be arbitrarily located on the complex plane. Then Proposition \ref{prop:equivalent2} directly follows from Proposition \ref{prop:equivalent}, Lemmas \ref{lem:distinct} and Theorem 1 
\end{proof}

Theorem 1 and Proposition \ref{prop:equivalent2} state that if the feedback system and the secure decoder are designed such that the conditions in Proposition \ref{prop:equivalent2} are satisfied, then we can guarantee accurate correction of $q$ errors using our proposed secure decoder for attacked nodes that can change over time, by checking the following very simple condition:
\begin{equation}
\forall v_i \in \mathbb{R}^n \text{ where } Av_i =\lambda_i v_i, ~ \lvert \textsf{supp}(Cv_i) \rvert > 2q.  \nonumber
\end{equation}
And interestingly, this is the exact same condition that one should check if one is designing the decoder from \cite{Fawzi2014} for fixed attack nodes. In other words, it is equally easy to check satisfiability of the sufficient condition for $q$-error-correction for both types of decoders.

%%%%%%%%%%%%%%%%%%%%%%%%%%%%%%%%%%%%%%%%%%%%%%%%%%%%%%%
%%%%%%%%%%%%%%%%%%%%%%%%%%%%%%%%%%%%%%%%%%%%%%%%%%%%%%%

\subsection{Discussion on Sufficient Condition of $T$}

Unsurprisingly, the condition on $T$ from Theorem 1, is different from that for a decoder designed for fixed attack nodes ($T\geq n$). In \cite{Fawzi2014}, since the set of attack nodes is fixed, one can leverage the property of fixed attacked nodes (i.e., the number of nonzero rows in (\ref{eq:opt_decoder})). However, in our setting, since the set of attack nodes can change over time, we cannot leverage this property and thus, we need more time steps $T$.
%to ensure accurate decoding is different depending on whether the attacked nodes are fixed or changing. 
In general, it is difficult to see what the value of $T$ is, from the formula in Theorem 1. However it is reasonable to assume that one would design the feedback matrix $G$ and matrix $C$ to maximize the number of errors that can be corrected, i.e. $q = \max q = \lceil p/2 - 1 \rceil$ (proved in Section \ref{sec:max_q}). To achieve this, we must design $G$ and $C$ such that $s_i = \lvert \textsf{supp} (C v_i) \rvert = p$ for all eigenvectors of $A$. In other words, $\operatorname{max} \mathcal{S}_m = \operatorname{min} \mathcal{S}_m = p$ for all $m$ and for all subsets $\mathcal{S}_m$. 
In this case, it is easy to see that according Theorem 1, for any $p$ that is an even number, we must choose: $T\geq T_n > (n-1) \cdot p/2$.
%\yh{ $2q = \lceil p - 2 \rceil = p-2 $  $\implies \max \mathcal{S}_m -2q = p - 2q = 2$. So, isn't it $T > (n-1)\cdot p /2 $ (no matter what $p$ is)?}
%\begin{equation}
%	T > \begin{cases}
%			(n-1) \cdot p/2 \quad &\text{ $p$ is even} \\
%                        	(n-1)\cdot p   \quad &\text{ $p$ is odd}.
%	\end{cases}
%	\nonumber
%\end{equation}
In most cases, this $T$ is larger than $n$ (the value of $T$ required for the decoder for fixed attacked nodes).

At the first look, this may seem like an unsatisfying result, however there are two important points to note. First, a larger value of $T$ merely translates to a longer initial delay from time $t = 0$ to $t=T-1$ when the decoder collects enough ($T$) measurements. From time $t=T-1$ onwards, the decoding is in real time and takes place at every time step $t$ as the decoder uses a ``sliding window'' of observations. Second, extensive simulation results in this paper show that a value of $T = n$ is often sufficient for our proposed decoder to achieve good estimation, i.e. to perfectly recover the state under attacks where the attacked nodes change over time. 

This is because we consider the worst case (conservative) scenario in Theorem 1. 
%\qie{Shall we leave the descriptions below here? Eigenvalues of A are now positive real. Should we add something on negative real eigenvalues too?} \yh{I do not think so. Even our simulation includes complex eigenvalues.} 
In the proof, we consider the case where $A$ has $n$ distinct positive eigenvalues. What if $A$ has complex eigenvalues? For instance, assume that $A \in \mathbb{R}^{3\times 3}$ $(n=3)$ and it has one pair of complex conjugate eigenvalues and one real eigenvalue, i.e., $\lambda_1, \lambda_2 (= \bar \lambda_1), \lambda_3$ where $\lambda_1, \lambda_2 \in \mathbb{C}$, $\lambda_3 \in \mathbb{R}$ and $\bar \cdot$ represents the complex conjugate. We denote $v_1 = x+ i y$, $v_2 = \bar v_1 = x - i y$ and $v_3 = w$ where $x$, $y$, $w  \in \mathbb{R}^n $ are linearly independent. Any $z \in \mathbb{R}^n$ can be represented by a linear combination of $n$ independent vectors in $\mathbb{R}^n$, i.e., $z = \alpha v_1 + \bar \alpha  v_2 + \beta w  = 2 Re( \alpha v_1) + \beta w =  \alpha_1 x + \alpha _2 y + \beta w$ where $\alpha \in \mathbb{C}$, $\beta \in \mathbb{R}$ and $\alpha_1 = Re(\alpha)\in \mathbb{R}$ and $\alpha_2 = - Im(\alpha) \in \mathbb{R}$. Therefore, the same results for real eigenvalues applies. In other words, if $ \lvert \textsf{supp}( C x) \rvert = \lvert \textsf{supp}( C y) \rvert = \lvert \textsf{supp}( C w) \rvert = p$ and $T > \frac { \sum_{k=0}^{T-1} s_{r,123}^k} {p - 2q}$, then $\lvert \textsf{supp} (\Phi z) \rvert > 2 q \cdot T$ where $s_{r,123}^k $ represents the number of cancelled support of linear combinations of $x,y$ and $w$ at time step $k$.

%%%%%%%%%%%%%%%%%%%%%%%%%%%%%%%%%%%%%%%%%%%%%%%%%%%%%%%
%%%%%%%%%%%%%%%%%%%%%%%%%%%%%%%%%%%%%%%%%%%%%%%%%%%%%%%

\subsection{Number of Correctable Errors}\label{sec:max_q}

Given that the set of attacked nodes can change over time and $e^{(t)}$ satisfies $\lvert \textsf{supp} (e^{(t)}) \rvert \le q$ for all $t$, we prove in Proposition \ref{prop:maximum} (see below) that the maximum number of correctable errors (as defined in Definition \ref{def:num_err_change}) by our decoder is $\lceil p/2-1 \rceil$, where $p$ is the number of measurements. This is in fact the same as the maximum number of correctable errors for the decoder proposed in \cite{Fawzi2014} which is for fixed attacked nodes.
This is a pleasing result, because it demonstrates that with our proposed decoder, we can relax the assumption of fixed attacked nodes and protect the system against more general attacks, without compromising the maximum number of correctable errors. 

%%%%%%%%%%%%%%%%%%%%%%%%%%%%%%%%%%%%%%
% LEMMA
%%%%%%%%%%%%%%%%%%%%%%%%%%%%%%%%%%%%%%

\begin{prop}\label{prop:maximum} 
Let $A_0 \in \mathbb{R}^{n \times n}$, $B \in \mathbb{R}^{n \times m}$ and $C \in \mathbb{R}^{p \times n}$ and assume that the pair ($A_0$, $B$) is controllable, $C$ is full rank and each row of $C$ is not identically zero. Then there exists a finite set $F \subset \mathbb{C}$ such that for any choice of $n$ numbers $\lambda_1, \cdots, \lambda_n \in \mathbb{C} \backslash F$ such that $0<\lambda_1< \cdots < \lambda_n$, there exists $G \in \mathbb{R}^{m \times n}$ such that:
\begin{itemize}
\item
The eigenvalues of the closed-loop matrix $A~(= A_0+BG)$ are $\lambda_1, \cdots, \lambda_n$.
\item
If the pair ($A, C$) is observable, then the number of correctable errors for the pair ($A, C$) is maximal after $T= \max\{n, T^*\}$ time steps and is equal to $\lceil p/2-1 \rceil$, where $T^*$ is the value of $T$ from Theorem 1. 
\end{itemize}
\end{prop}

\begin{proof}
The proof for Proposition 4 in \cite{Fawzi2014} shows that if the chosen poles $\lambda_1, \cdots, \lambda_n$ are distinct, positive and do not fall in some finite set $F$, then there is a choice of $G$ such that the eigenvalues of $A~(=A_0+B)$ are exactly the $\lambda_1, \cdots, \lambda_n$, and the corresponding eigenvectors $v_i$ are such that $\lvert \textsf{supp} (C v_i) \rvert = p$. Thus, by Proposition \ref{prop:equivalent2}, the number of correctable errors for $(A,C)$ is $\lceil p/2-1 \rceil$.
\end{proof}

%\begin{proof} By Proposition \ref{prop:qT_error}, if all subsets of $2q\cdot T$ columns of $Q_2^\top \in \mathbb{R}^{(p \cdot T - n  ) \times p\cdot T}$ are linearly independent, then $x^{(0)}$ is the unique solution of (\ref{eq:direct_l1}) and hence $q$ non-fixed errors are correctable. Since the number of rows of $Q_2 ^ \top$ is smaller than the number of columns of $Q_2 ^ \top$ and all subsets of $2q \cdot T $ columns of $ Q_2  ^\top$ are linearly independent, we must have $p\cdot T - n \ge 2 q \cdot T$ to satisfy the condition of Proposition 3. Rearranging this gives $p/2 - n/(2T) \geq q$. Because $p, q \in \mathbb{N}$, hence after $T=n$ steps, $ p /2 - 1/2 \ge q $, i.e., the maximal number of errors that could be corrected is $q = \lfloor p/2 - 1/2 \rfloor =  \lceil p/2 - 1 \rceil$.
%Finally, since the pair $(A,C)$ and the value of $T$ satisfy the conditions in Theorem 1, then this many errors are correctable.
%\end{proof}

%\noindent In \cite{Fawzi2014}, the authors prove the maximum number of correctable fixed attacked nodes using (Lebesgue) measure zero. However, we prove the maximum number of correctable nodes, when they can change over time, using the dimensionality of $Q_2^\top$ and Theorem 1. 

In addition, recall that $E_{q,T}$ consists of vertically stacked error vectors from $e^{(0)}, \cdots, e^{(T-1)}$, and observe that our proofs for accurate decoding are independent of how the individual error (nonzero) terms are distributed in the vector $E_{q,T}$. Thus, if we remove the assumption: $\lvert \textsf{supp} (e^{(t)}) \rvert \le q$ for all $t$, and allow $e^{(t)}$ to appear in an arbitrary fashion, e.g. $\lvert \textsf{supp} (e^{(0)}) \rvert = 0$ and $\lvert \textsf{supp} (e^{(1)}) \rvert = 2q$, as long as $\sum_{t=0}^{T-1} \lvert \textsf{supp} (e^{(t)}) \rvert \leq q\cdot T$, then our $q$-error-correcting decoder can still recover the true states. In other words, our proposed decoder can protect the system against more general attacks, where the number of attacked nodes is not necessarily less than or equal to $q$ at every time.

\begin{figure*}[!t]
\center
\includegraphics[width=0.85\textwidth]{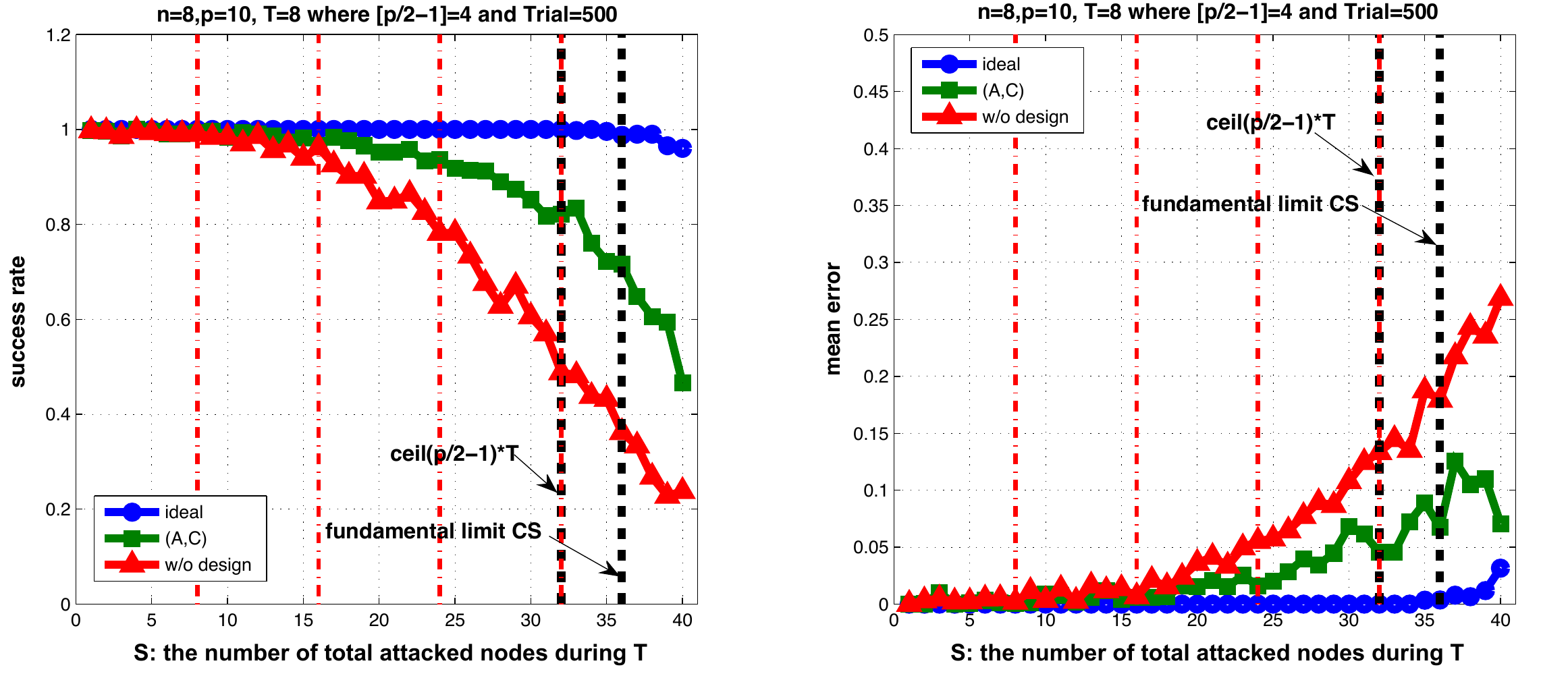}
\caption{Success rate and mean error of $l_1$ decoder on different systems (ideal coding matrix, designed state feedback and poorly designed system with $n=8$, $p=10$ and $T=n$, where black dot lines show the fundamental limit for dynamical systems and ideal coding matrix case respectively. We see that as the number of attacked nodes increase, success rate decreases. Also, by designing state feedback gain properly, we improve success rate and decrease mean error. }
\label{fig:ex_n8p10}
\end{figure*}
%%%%%%%%%%%%%%%%%%%%%%%%%%%%%%%%%%%%%

\section{Optimal Decoder Design}\label{sec:design}
 
In the classical error correction problem, to ensure accurate decoding, the coding matrix must satisfy the RIP conditions  \cite{Candes_Tao}, which are extremely difficult to check in general. In practice, Theorem 1.4 from\cite{Candes_Tao} is almost always used to design a coding matrix {\it a priori}. This theorem states that a coding matrix whose entries are sampled from independent and identical distributions satisfies the RIP condition with overwhelming probability. 
In secure estimation, however, it is impossible to choose such a coding matrix {\it a priori} because it is the observability matrix $\Phi$, which is structurally constrained: as shown in (\ref{eq:decoder_Phi}), $\Phi$ consists of $CA^{i}$'s where $i=\{0, \cdots, T-1\}$. In this Section, we use Lemma \ref{lem:distinct}, the results from Proposition 4 in \cite{Fawzi2014} and state feedback to design a matrix $\Phi$ for accurate decoding.

%%%%%%%%%%%%%%%%%%%%%%%%%%%%%%%%%%%%%%%%%%%%%%%%%%%%%%
%%%%%%%%%%%%%%%%%%%%%%%%%%%%%%%%%%%%%%%%%%%%%%%%%%%%%%

\subsection{System and Decoder Design}\label{sec:decoder_design}

Since conditions $(i)$ and $(iii)$ in Theorem 1 are equivalent, condition $(i)$ can be used to design a state feedback controller such that the closed system can achieve accurate decoding. Therefore, given a controllable open-loop pair $(A_0,B)$, design $C$ and choose an adequate feedback control law $u^{(t)} = G x^{(t)}$ and construct a secure decoder such that:
\begin{enumerate}
	\item Each row of $C$ is not identically zero, and $C$ is full rank;
	\item The closed-loop matrix $A~(=A_0+BG)$ has $n$ distinct positive eigenvalues: $0< \lambda_1 < \lambda_2 < \cdots < \lambda _n$;
	\item $(A,C)$ is observable;
	\item The length of the sliding window of measurements $T$ of the decoder satisfies Theorem 1\footnote{We found that much smaller $T$'s are often sufficient for good secure estimation performance, i.e. to perfectly recover the attack signals. In all simulations in this paper, $T=n$ is used, where $n$ is the number of states.};
	\item Maximize $q$ subject to: $\forall v_i \in \mathbb{R}^n$ where $Av_i = \lambda v_i$, $\lvert \textsf{supp} (C v_i)  \rvert > 2q$.
\end{enumerate}
Without loss of generality, the first condition holds. For example, if there exists a zero row in $C$, we can simply remove that row from $C$ without changing the system's behavior. Conditions 2, 3 and 4 are required for equivalence in Theorem 1. The last condition is needed for accurate decoding and for maximizing the number of correctable attacks. From Proposition \ref{prop:maximum} the maximum number of correctable errors can be achieved when $\lvert \textsf{supp} (C v_i)  \rvert = p$ (i.e., the number of measurements) for all eigenvectors of $A$. 

Conditions 2, 3 and 5 depend on the feedback controller. So how do we choose a controller that achieves good performance in both control and secure estimation? Below, we describe an approach that we have taken in all simulations in this paper, and has proved to work well. This is by no means the only method. First, we design a controller that achieves good control, for example, Linear Quadratic Regulators (LQR), which are optimal with respect to a certain quadratic cost function. However, these controllers may not have good secure estimation properties, meaning the value of $q$ that satisfies $\lvert \textsf{supp} (C v_i) \rvert > q$ for all eigenvectors of $A$ may be small, i.e., the resulting decoder can only correct few errors. It is often easy to increase the value of $q$ and make the decoder more resilient to attacks by slightly perturbing the closed-loop poles from those resulting from the LQR controller, such as placing the poles closer to the origin, and making the poles more spread out amongst themselves. We chose to keep the perturbations small as to not lose too much control performance.  
Although this is a heuristic method, it is relatively easy to carry out in order to satisfy the above conditions; whereas in the classical error correction method \cite{Candes_Tao}, checking whether a coding matrix satisfies RIP is extremely difficult.

To summarize, we start from some optimal controller which may not result in a good decoder, then we perturb the closed-loop poles slightly to improve the resulting decoder's secure estimation capability. Therefore there is a trade-off between a system's control and secure estimation performances, and the feedback controller can be designed to achieve a desired trade-off between them.

%%%%%%%%%%%%%%%%%%%%%%%%%%%%%%%%%%%%%
\begin{figure}
\center
\includegraphics[width=0.45\textwidth]{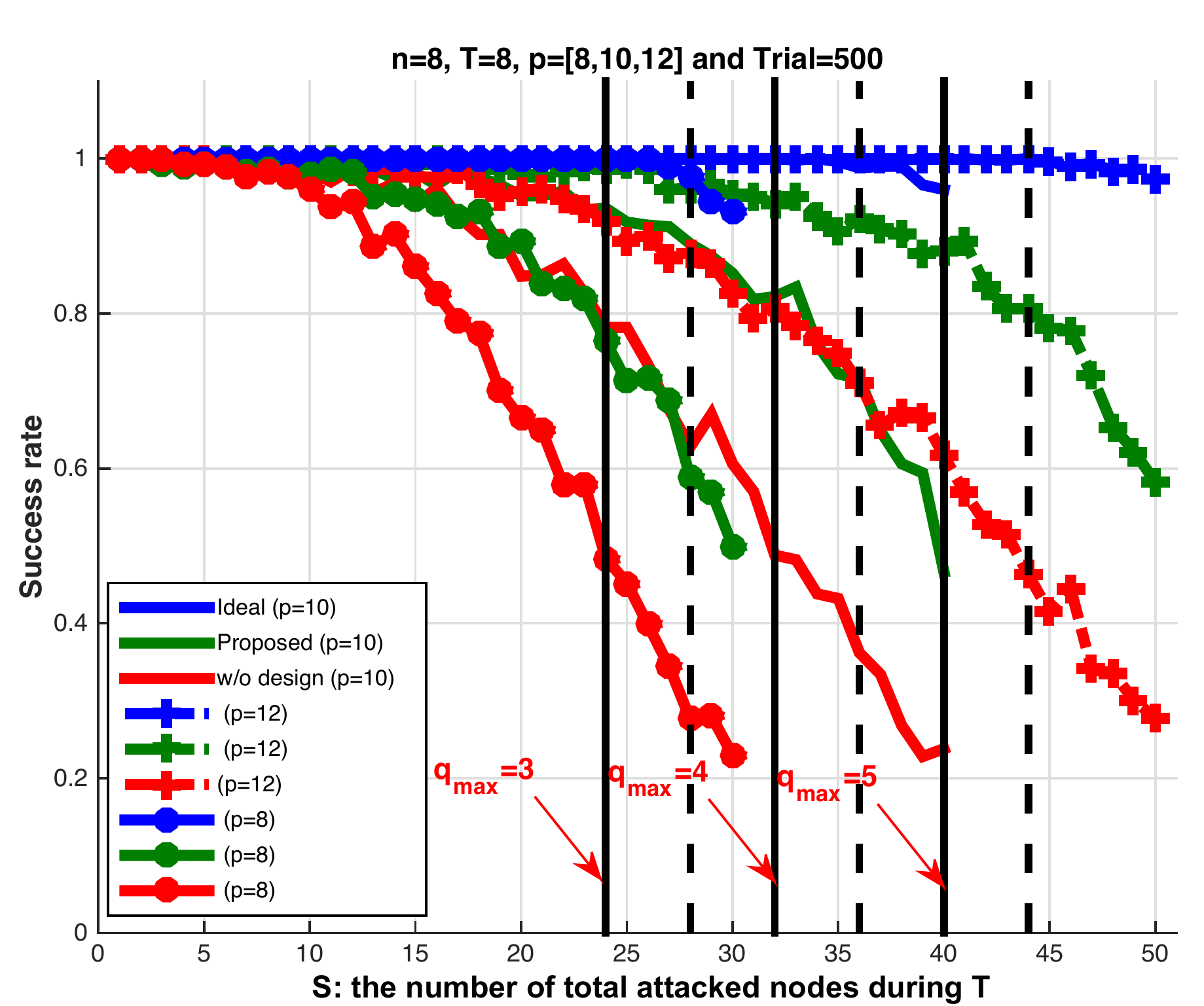}
\caption{Success rate and mean error of $l_1$ decoder on three different systems (ideal coding matrix, designed state feedback and poorly designed system with $n=8$ and $T=n$) with different $p=[8,10,12]$. Black solid lines show the fundamental limit for dynamical systems and black dashed lines show the fundamental limit for the ideal coding matrix case. We see that as the number of attacked nodes increases, the success rate decreases. Also, by designing the state feedback gain properly, we improve success rate and decrease mean error. }
\label{fig:ex_n8_overall}
\end{figure}
%%%%%%%%%%%%%%%%%%%%%%%%%%%%%%%%%%%%%

%%%%%%%%%%%%%%%%%%%%%%%%%%%%%%%%%%%%%%%%%%%%%%%%%%%%%%
%%%%%%%%%%%%%%%%%%%%%%%%%%%%%%%%%%%%%%%%%%%%%%%%%%%%%%

\subsection{Practical Performance}\label{sec:prac_perf}
In this Section, we show the performance of our proposed decoder using an arbitrary system with $n=8$, $p=10$ and $T=n$, where $A$ is chosen arbitrarily (i.e., the entries of $A$ are chosen as independent and identically distributed random variables) and $C$ is chosen such that every row of $C$ has only one nonzero component (i.e., each row of $C$ is not identically zero) and is full rank. 
%This setting is more reasonable than a random system \cite{Fawzi2014}, i.e., i.i.d. Gaussian entries, since such a matrix would have good RIP. 
For each value of $S$ (i.e. number of total attacked nodes during $T$ steps), we test the decoder on 500 independent trials. For each trial, both the system matrices and initial conditions are re-generated. The initial conditions $x^{(0)}$ were randomly generated from the standard Gaussian distribution and the set of attacked nodes are chosen at random and can change over time. Since we increase $S$ in steps of $1$ as shown in Figure \ref{fig:ex_n8p10}, %\qie{`attacted' should be `attacked' in Figure 1 and 2's x-axis label. Is it easy to change? If not, let's just forget about it.}, 
we first distribute $q= \lfloor S/T \rfloor$ attacks arbitrary for each time step $t$ and randomly distribute the remaining $(S-q\cdot T)$ attacks amongst $e^{(0)},\cdots, e^{(T-1)}$. For example, in Figure \ref{fig:ex_n8p10}, for $S=20$, we have at least 2 attacks for each $e^{(t)}$ where $t=\{0,...,7\}$ and 
distribute the remaining $4$ attacks arbitrary. 

Figures \ref{fig:ex_n8p10} and \ref{fig:ex_n8_overall} show the performance of our proposed decoder on three different systems: (1) an ideal random coding matrix (i.e., with i.i.d. Gaussian entries), (2) a LTI system with a state feedback controller designed using the method in Section \ref{sec:decoder_design} and (3) a LTI system with a poorly designed feedback controller. 
The performance metric used is the \textit{success rate}, which is defined as the fraction of the total 500 trials that our decoder is able to perfectly recover the attack signals.
System 1 acts as the baseline for the comparison because a random coding matrix has good RIP and represents the best achievable decoding performance. However it is not a realistic dynamical system. For a realistic LTI system, the coding matrix $\Phi$ is constructed by stacking $CA^{t}$'s as in (\ref{eq:decoder_Phi}) where $t=\{ 0, ..., T-1\}$. This structural constraint can reduce the restricted isometry constant (i.e., a measure of how good the RIP is, larger restricted isometry constants correspond to better RIP), and indeed, the success rates for system 3 are lower than those for system 1. Figure 1 also shows that by designing the feedback controller, we can recover some of the lost RIP due to the structure in the system dynamics, and improve the success rate (system 2).
Fundamental limits of the $l_1$-optimization for both secure decoding and ideal coding are also shown in Figure \ref{fig:ex_n8p10}.

Figure \ref{fig:ex_n8_overall} shows that as $S$ increases, more measurements ($p$) are needed to correctly recover the state of the system. In practice, this can be done by fusing different types of measurements and sensors together. For example, consider a simplified UAV dynamics where $ x = \begin{bmatrix} p_x, p_y, p_z, v_x, v_y, v_z \end{bmatrix}^\top$ and $p_{(\cdot)}$'s and $v_{(\cdot)}$'s represent positions and velocities, respectively. The observation equations of the IMU and GPS are as follows:
\begin{equation}
C_{IMU} =\begin{bmatrix} {\bf 0}_{3\times 3} & {\bf I} \end{bmatrix}, C_{GPS} = \begin{bmatrix} {\bf I} & {\bf 0}_{3\times 3} \end{bmatrix} \nonumber
\end{equation}
By combining measurements from both IMU and GPS, we can increase the types of measurements ($p$), $C = \begin{bmatrix} C_{IMU} ; C_{GPS}\end{bmatrix} \in \mathbb{R}^{6 \times n}$.

Finally, we want to point out that $T=n=8$ in all these simulations, which is much smaller than that dictated by Theorem 1 ($T=35$).
Nevertheless, our decoder's success rates are quite high: in Figure \ref{fig:ex_n8p10}, it achieves success rates of more than 0.8 for system~2 for all attacks where the number of attacked nodes is below the maximum number of correctable errors, in Figure \ref{fig:ex_n8_overall}, these success rates increase to above 0.9 when 12 measurements are available. These results suggest that the sufficient condition on $T$ given in Theorem 1 is conservative, and can be relaxed in practice.

%%%%%%%%%%%%%%%%%%%%%%%%%%%%%%%%%%%%%
\begin{figure}
\center
\includegraphics[width=0.45\textwidth]{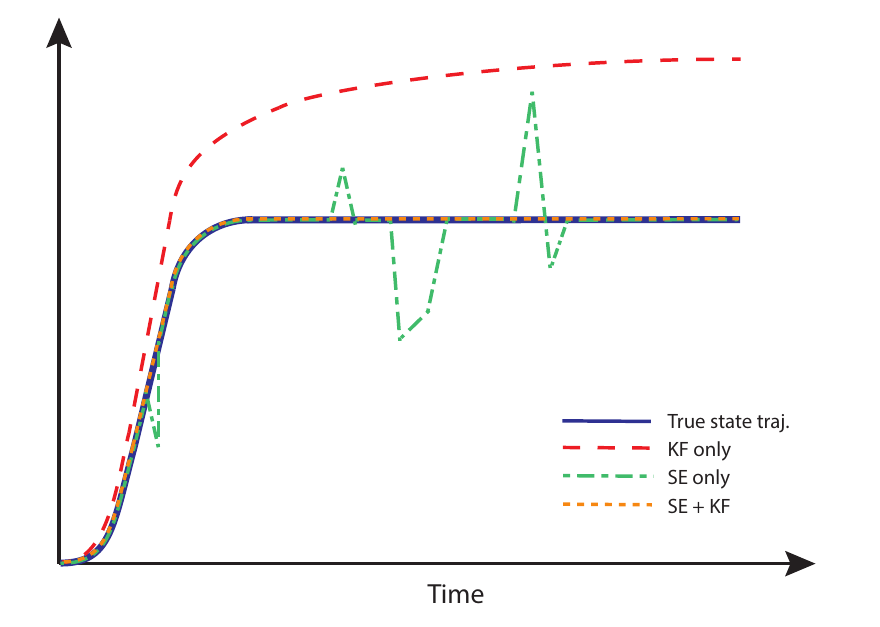}
\caption{Illustrative comparison of three schemes: KF only (KF), secure estimator only (SE), secure estimator with a KF (SE+KF). KF fails to estimate the true state as attack signal is non-Gaussian. SE correctly estimates the system state most of the time but has occasional large estimation errors. SE+KF tracks the true state trajectory perfectly.}
\label{fig:estimation}
\end{figure}
%%%%%%%%%%%%%%%%%%%%%%%%%%%%%%%%%%%%%

%%%%%%%%%%%%%%%%%%%%%%%%%%%%%%%%%%%%%%%%%%%%%%%%%%%%%%
%%%%%%%%%%%%%%%%%%%%%%%%%%%%%%%%%%%%%%%%%%%%%%%%%%%%%%

\subsection{Combination of Secure Estimation and Kalman Filter}\label{sec:estimation}
Consider the state estimation problem for the following LTI system under attack:
\begin{equation}
\begin{aligned}
x^{(t+1)} & = A x^{(t)} + B u^{(t)} + w^{(t)}\\
y^{(t)} & = C x^{(t)} + e^{(t)} + v^{(t)}
\end{aligned}
\end{equation}
where $x^{(t)}$, $y^{(t)}$, $u^{(t)}$ and $e^{(t)}$ are as defined in (\ref{eq:system_model_se});
%$x$ is the state, $u$ is the control input, $y$ is the output, $e$ is the attack signal, 
$w$ and $v$ are zero mean i.i.d. Gaussian process noise and measurement noise, respectively. 

We can model the attack signal as noise and use a KF to estimate the states. More specifically, one would define a new measurement noise vector $\tilde{v}^{(t)} = e^{(t)} + v^{(t)}$, so that the measurement equation becomes $y^{(t)} = C x^{(t)} + \tilde{v}^{(t)}$. A KF can then estimate the states from the inputs $u^{(t)}$ and the corrupted measurements $y^{(t)}$ \cite{KwonACC}. However, KFs are derived using the assumption that both process and measurement noises are zero mean i.i.d. white Gaussian processes, but attack signals are usually erratic and may be poorly modeled by Gaussian processes \cite{KwonACC}. For example, in GPS spoofing attacks, attack signals are often structured to resemble normal GPS signals or can be genuine GPS signals captured elsewhere. %Such signals are neither Gaussian nor white. 
When the system is subjected to attacks that are poorly modeled by Gaussian processes, a KF is expected to fail to recover the true states. Figure \ref{fig:estimation} gives an illustrative example where an attack signal that increases linearly with time is injected into the measurements of state $x_i$. The red dashed line shows a plausible estimated state trajectory from a KF.

On the other hand, our proposed secure decoder works for arbitrary and unbounded attacks.
Nevertheless, Figures \ref{fig:ex_n8p10} and \ref{fig:ex_n8_overall} show that when the number of attacked nodes are close to the theoretical maximum number of correctable errors, our decoder occasionally fails to perfectly recover the states.
%\st{such state estimates may not be satisfactory as secure estimation fails perfect recovery as described in Section,} 
The green dashed line in Figure \ref{fig:estimation} depicts a possible result from this decoder: the estimated state trajectory follows the true trajectory most of the time with occasional errors.

Therefore, to improve the performance, we propose to combine our secure decoder with a KF as follows:
\begin{algorithm}
\caption{Combined secure estimator with KF}
\label{al:se_kf}
\begin{algorithmic}[1]
\State Initialize the KF
\For{each $t$}
	\If{$t \geq T$}
		\State Estimate the attack signal at time $t$, $\hat e^{(t)}$, using secure estimator
	\Else
		\State Set $\hat e^{(t)} = 0$
	\EndIf
	\State Form a new measurement equation: $\tilde y^{(t)} =  C x^{(t)} + \tilde v^{(t)}$, where $\tilde y^{(t)} = y^{(t)} - \hat e^{(t)}$ and $ \tilde v^{(t)} = e ^{(t)} - \hat e^{(t)} + v^{(t)}$
	\State Apply standard KF using $u$ and $\tilde y$ 
\EndFor
\end{algorithmic}
\end{algorithm}

\noindent
The intuition is that the secure decoder acts as a pre-filter for the KF, so that $\tilde v^{(t)}$ is close to a zero mean i.i.d. Gaussian process even when the true attack signal $e^{(t)}$ is not. More specifically, the secure decoder usually perfectly recovers $e^{(t)}$, thus $e^{(t)} - \hat e^{(t)} = 0$ and $\tilde v^{(t)} = v^{(t)}$. What happens when the secure decoder fails? Equation (\ref{eq:sys_err_corr}) shows that the estimated state at time $t$, $\hat x^{(t)}$, %depends on measurements from the past $T$ steps, however it 
does not directly depend on the estimated state at another time point $\hat x^{(\tau)}$ ($t \neq \tau$). As a result, when the secure decoder fails, its estimation error, $e^{(t)} - \hat e^{(t)}$, appears to be quite random. Putting these together: $\tilde v^{(t)} = e^{(t)} - \hat e^{(t)} + v^{(t)}$ is closer to a zero mean white Gaussian process than the original attack signal $e^{(t)}$, and this improves the KF's performance. 

Finally, the \textit{if} statement in Algorithm 1 ensures that the secure estimator always has access to $T$ past measurements, as required by Theorem 1.
A more realistic example illustrating these behaviors is shown in Figure \ref{fig:ex_uav_remote}.

% !TEX root = IEEE_adversarial_attacks.tex
\section{Numerical Examples}\label{sec:examples}

On February 15, 2015, the Federal Aviation Administration proposed to allow routine use of certain small, non-recreational UAVs in today's aviation system \cite{faa}. Thus in the near future, we may see thousands of UAVs such as Amazon Prime Air \cite{Amazon} and Google Project Wing vehicles \cite{Google} sharing the airspace simultaneously. To ensure safety of this immense UAV traffic, UAVs may periodically update their position and velocity measurements wirelessly to a Remote Control Center (RCC) for traffic management (Channel 1 in Figure \ref{fig:ex_uav_pic}). At the same time, UAVs may broadcast this information to other UAVs in its vicinity for collaborative collision avoidance (Channel 2 in Figure \ref{fig:ex_uav_pic}). Finally, autonomous UAVs may use GPS for their position measurements (Channel 3 in Figure \ref{fig:ex_uav_pic}). 
All these communication channels are subject to cyber attacks. 
If corrupted information are used in collision avoidance or path planning algorithm, they can lead to possible collisions or loss of UAVs, causing physical and financial damage and even injury to civilians.
To help protect against these attacks and consequences, participating entities such as the UAVs and the RCC can use secure estimation to estimate a target UAV's true position and velocity before using any received information for collision avoidance, for instance.
In this section, we focus on 2 types of adversarial cyber attacks on UAVs and demonstrate the effectiveness of our secure estimator through simulations.

%%%%%%%%%%%%%%%%%%%%%%%%%%%%%%%%%%%%%
\begin{figure}
\center
\includegraphics[width=0.35\textwidth]{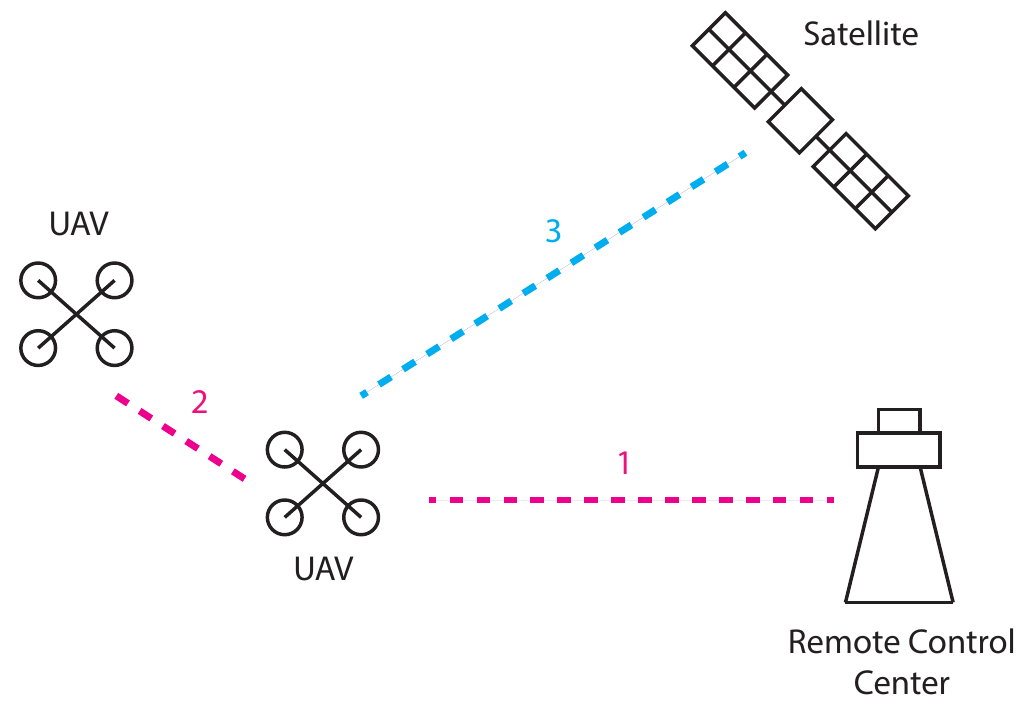}
\caption{Different communication channels that are subject to adversarial attacks.}
\label{fig:ex_uav_pic}
\end{figure}
%%%%%%%%%%%%%%%%%%%%%%%%%%%%%%%%%%%%%

%%%%%%%%%%%%%%%%%%%%%%%%%%%%%%%%%%
\begin{figure}[b]
\center
\includegraphics[width=0.40\textwidth]{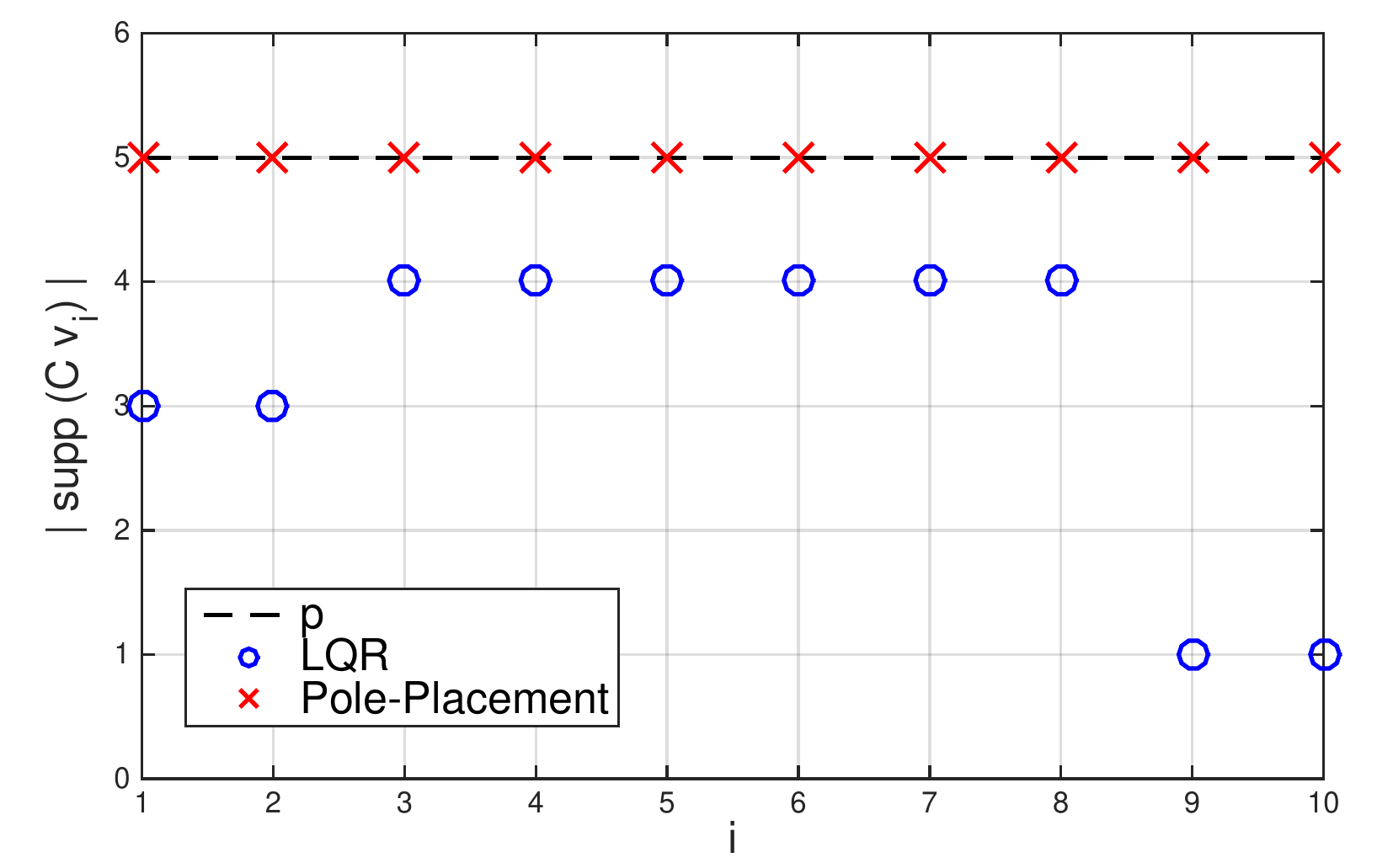}
\caption{$ \lvert \textsf{supp} (C v_i) \rvert $ for all eigenvectors $v_i$ of closed-loop matrix $A$ for 2 feedback controllers: a LQR and a controller designed by pole-placement. Black dashed line is at $p = 5$, i.e., the number of measurements.}
\label{fig:ex_pole}
\end{figure}

%%%%%%%%%%%%%%%%%%%%%%%%%%%%%%%%%%
\subsection{UAV Model}
We consider a quadrotor with the following dynamics:
\begin{equation}
\begin{aligned}
x^{(t+1)} &= A_0 x^{(t)} + B u^{(t)}  + k + w^{(t)} \\
y^{(t)} &= C x^{(t)} + e^{(t)} + v^{(t)} \\
\end{aligned}
\end{equation}
\begin{figure*}
\begin{equation}
\begin{aligned}
A_0 &= \begin{bmatrix} 1 & T_s &  g \cdot \frac {T_s^2} {2} & 0 		& 0 & 0 & 0 & 0 	& 0 &   0 \\
			      0 & 1 & g T_s & 0 				& 0 & 0 & 0 & 0 	& 0 &   0 \\
			      0 & 0 & A_\theta^{11} & A_\theta^{12} 			& 0 & 0 & 0 & 0 	& 0 &   0 \\
			      0 & 0 & A_\theta^{21} & A_\theta^{22}  			& 0 & 0 & 0 & 0 	& 0 &   0 \\
			      0 & 0 & 0 & 0 & 1 & T_s &  g\frac{T_s^2} {2} & 0	& 0 & 0 \\
			      0 & 0 & 0 & 0 & 0 & 1 & g \cdot T_s & 0  & 0 & 0 \\
			      0 & 0 & 0 & 0 & 0 & 0 &  A_\theta^{11} & A_\theta^{12} & 0 & 0 \\
			      0 & 0 & 0 & 0 & 0 & 0 &  A_\theta^{21} & A_\theta^{22}  & 0 & 0 \\
			      0 & 0 & 0 & 0 & 0 & 0 & 0 & 0 & 1 & T_s \\
			      0 & 0 & 0 & 0 & 0 & 0 & 0 & 0 & 0 & 1
	\end{bmatrix}, 
B = \begin{bmatrix}  
				0	& 	0	& 	0 \\
				0	& 	0	& 	0 \\
				B_\theta^1 & 0 & 0 \\
				B_\theta^2 & 0 & 0 \\
				0	& 	0	& 	0 \\
				0	& 	0	& 	0 \\
				0 &  B_\theta^1  & 0 \\
				0 &  B_\theta^2  & 0 \\
				0 & 0 & \frac {k_T \cdot T_s^2} {2m} \\
				0 & 0 & \frac {k_T \cdot T_s} {m}
	\end{bmatrix}, \\
%k = \begin{bmatrix}
%				0\\
%				0\\
%				0\\
%				0\\
%				0\\
%				0\\
%				0\\
%				0\\
%				- \frac {g \cdot T_s^2} {2m}\\
%				- \frac {g T_s} {m}
%	\end{bmatrix}\\
%C_I &= \begin{bmatrix} 
%			1 & 0 & 0 & 0 & 0 & 0 & 0 & 0 & 0 \\
%			0 & 0 & 0 & 0 & 1 & 0 & 0 & 0 & 0 \\
%			0 & 0 & 0 & 0 & 0 & 0 & 0 & 1 & 0 
%	\end{bmatrix}
\end{aligned}
\end{equation}
\end{figure*}
\noindent 
where $x = [p_x, v_x, \theta_x, \dot \theta_x, p_y, v_y, \theta_y, \dot\theta_y, p_z, v_z]^T$ is the state vector. $p_x$, $p_y$ and $p_z$ represent the quadrotor's position along the $x$, $y$ and $z$ axis, respectively. $v_x$, $v_y$ and $v_z$ represent its velocities. $\theta_x$ and $\theta_y$ are the pitch and roll angles respectively, $\dot \theta_x$ and $\dot \theta_y$ are their corresponding angular velocities. %\st{represents the quadrotor's position, velocity, pitch and roll angles and angular velocities.} 
$u = [\theta_{r,x}, \theta_{r,y}, F]^T$ is the input vector: $\theta_{r,i}$ is the reference pitch or roll angle, and $F$ is the commanded thrust in the vertical direction. $y = [\tilde{p}_x, \tilde{p}_y, \tilde{p}_z]^T$ %\st{$y = [p_x, p_y, p_z]^T$}
represents compromised position measurements from the GPS under attack signal $e$. $w$ and $v$ represent process and measurement noise respectively. $k$ is a constant vector which represents gravitational effects, and can be dropped without loss of generality because we can always subtract it out in $u$. $A_\theta^{i,j}$ refers to the $ij$-th entry of the subsystem matrix of the discretized rotational dynamics $A_\theta$, and $B_\theta^i$ refers to the $i$-th entry of the input-to-state map $B_\theta$ for the discretized rotational dynamics. $T_s$ is the discrete time step, $g$ is the gravitational acceleration, $m$ is the mass of the quadrotor and $K_T$ is a thrust coefficient. Further details about this model and its derivation can be found in \cite{Bouffard}. Finally, the matrix $C$ depends on the particular measurements taken in each example.

\subsection{Decoder Design via Pole-Placement}

Assume that the UAV uses the state feedback control law $u^{(t)} = G x^{(t)}$, where $G$ is the feedback matrix which can be designed\footnote{In the GPS spoofing example, direct uncorrupted state measurements are not available. Therefore a KF is used to give estimated states which are then used for state feedback control.}. If the pair $(A_0, B)$ is controllable, then we can choose $G$ to place the closed loop poles anywhere in the complex plane. We first design a Linear Quadratic Regulator (LQR) and evaluate its secure estimation performance by checking whether the sufficient condition for $q$-error correction  (i.e., $\lvert \textsf{supp} (C v_i) \rvert > q$ for all $i$) holds. 
Figure \ref{fig:ex_pole} shows the results for a matrix $C \in \mathbb{R}^{5\times 10}$ (i.e., 5 measurements) and observe that $\lvert \textsf{supp} (C v_i) \rvert < p = 5$ for $i=1,2,9$ and $10$. Furthermore, $\lvert \textsf{supp} (C v_i) \rvert = 1 > 0$ for $i = 9$ and $10$, therefore the resulting secure decoder can correct zero errors!
To improve the secure estimation performance, we perturb the closed-loop poles slightly until $\lvert \textsf{supp} (C v_i) \rvert = p$ for all $i$, as shown in Figure \ref{fig:ex_pole}. Therefore the resulting secure decoder can achieve the maximum number of correctable errors within the limits of $p$ (i.e., the number of measurements). By keeping the perturbations on the poles small, our final controller achieves both good control and estimation performances (see Figure \ref{fig:ex_pp_est}).% and \ref{fig:ex_pp_err}). 

\subsection{UAV under Adversarial Attack}

\subsubsection{MITM Attack in Communication with a RCC or with other UAVs} \label{sec:uav_utm}
In this section, we consider MITM attacks targeted at Channels 1 and 2 in Figure \ref{fig:ex_uav_pic}, where a malicious agent spoofs the information being sent and/or received over these channels. The goal of the RCC or other UAVs is to accurately estimate the true flight path of a target UAV from corrupted measurements. 
Note that the true path of the target UAV is unaffected by the attack.
Assume that the attacker spoofs the position measurements in order to deceive the receiver that the target UAV is deviating in the $x$-direction, i.e., he/she injects a continuous and increasing signal in the $x$-position measurement.
To make the estimation task even harder for the receiver, the attacker also injects a random Gaussian noise to an additional measurement, and the choice of this measurement can change at each time step. 

In this example, we first demonstrate the effectiveness of our proposed decoder design via pole-placement method by comparing the estimation performance of the decoder resulting from (1) a LQR controller and (2) a controller designed using pole-placement as described in the previous section.
We then implement the latter feedback controller, and compare the performance of three different state estimation schemes: (1) KF only (KF), (2) secure estimator only (SE), and (3) secure estimator combined with KF (KF+SE). 

Throughout this example, $y \in \mathbb{R}^5$, measurements include the $x$, $y$ and $z$ positions and 2 additional randomly selected states. The left plots in Figure \ref{fig:ex_pp_est} show the true attack signal on all 5 sensors (solid lines) and the estimated attack signals (dashed lines) by the secure decoder if the feedback controller is a LQ regulator (top) or one designed via pole-placement (bottom). It is obvious that the latter estimates the attack signal much more accurately. The right plots of this figure highlights this observation by explicitly showing the estimation error of the attack signal for each measurement.

The same information is shown in Figure \ref{fig:ex_pp_err}, where each row corresponds to one sensor, and the first 3 rows are the $x$, $y$ and $z$ position measurements, respectively. This figure highlights three points: first, the attacked sensors change with time; second, the number of attacked sensors at each time $t$ is less or equal to 2; third, only position measurements are corrupted.

Figure \ref{fig:ex_uav_remote} shows the estimated flight paths by all three methods.
The true path of the UAV (solid blue line) starts from the position marked by the blue triangle and ends at the position marked by the blue square. KF fails to filter out the attack signal in the $x$-position measurements as the attack is highly non-Gaussian, and the estimated trajectory (dashed red line) significantly differs from the true one. On the other hand, SE correctly estimates some portions of the trajectory and the final position of the vehicle, nevertheless it produces spontaneous errors in the $x$ direction. 
Finally the combined method KF+SE perfectly recovers the true path of the target UAV.

%%%%%%%%%%%%%%%%%%%%%%%%%%%%%%%%%%%%%%%%
\begin{figure}
\center
\includegraphics[width=0.45\textwidth]{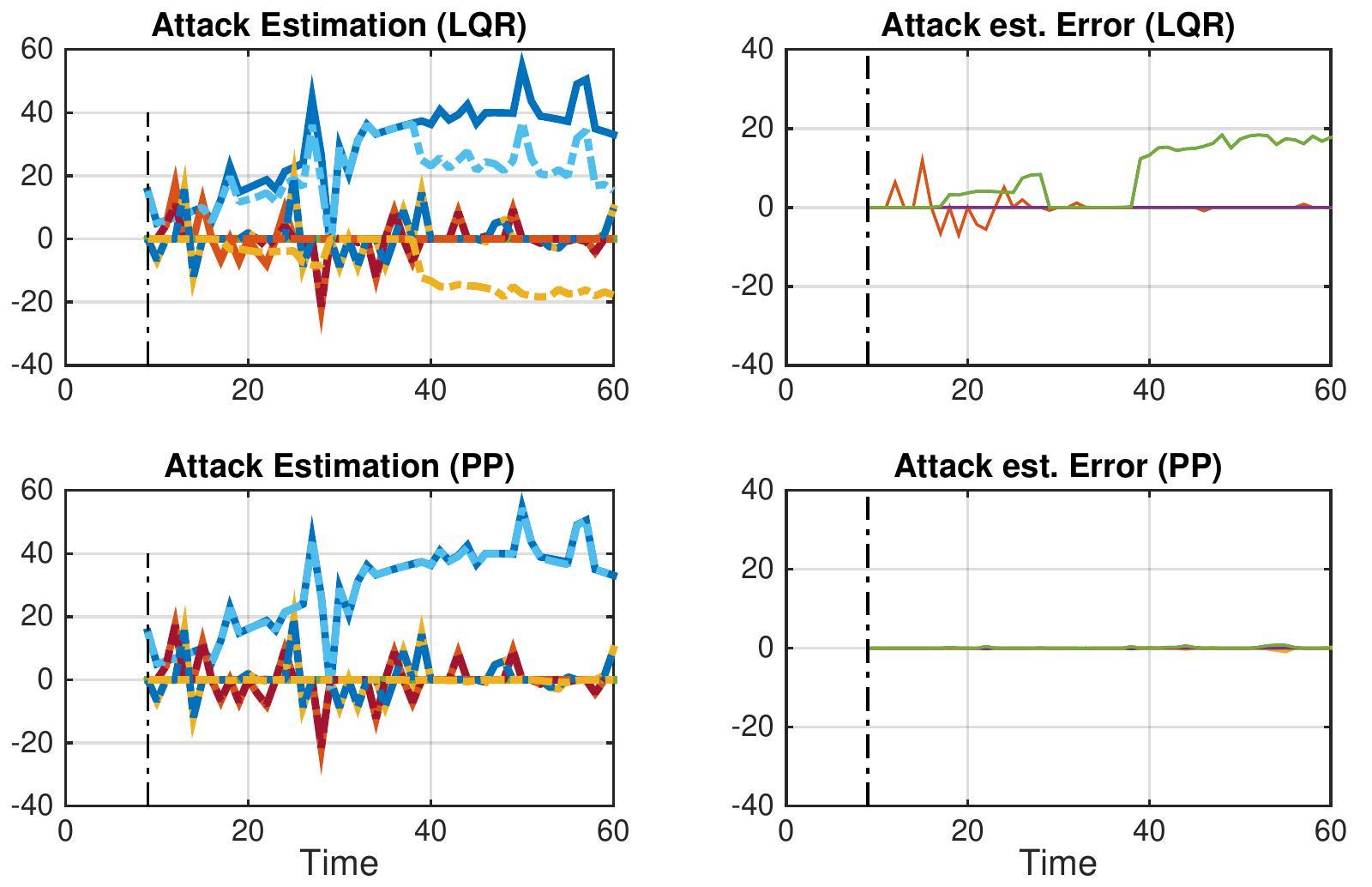}
\caption{True attack signal, estimated attack signal and estimation error in the attack signal of the estimator (SE) with 2 different feedback controllers: LQR, controller designed via pole-placement (PP); with 5 measurements. In the left plots, solid lines are true attack signals, dashed lines are estimated signals. The right plots show the estimation error in the attack signal.}
\label{fig:ex_pp_est}
\end{figure}
%%%%%%%%%%%%%%%%%%%%%%%%%%%%%%%%%%%%%

%%%%%%%%%%%%%%%%%%%%%%%%%%%%%%%%%%%%%
\begin{figure}
\center
\includegraphics[width=0.48\textwidth]{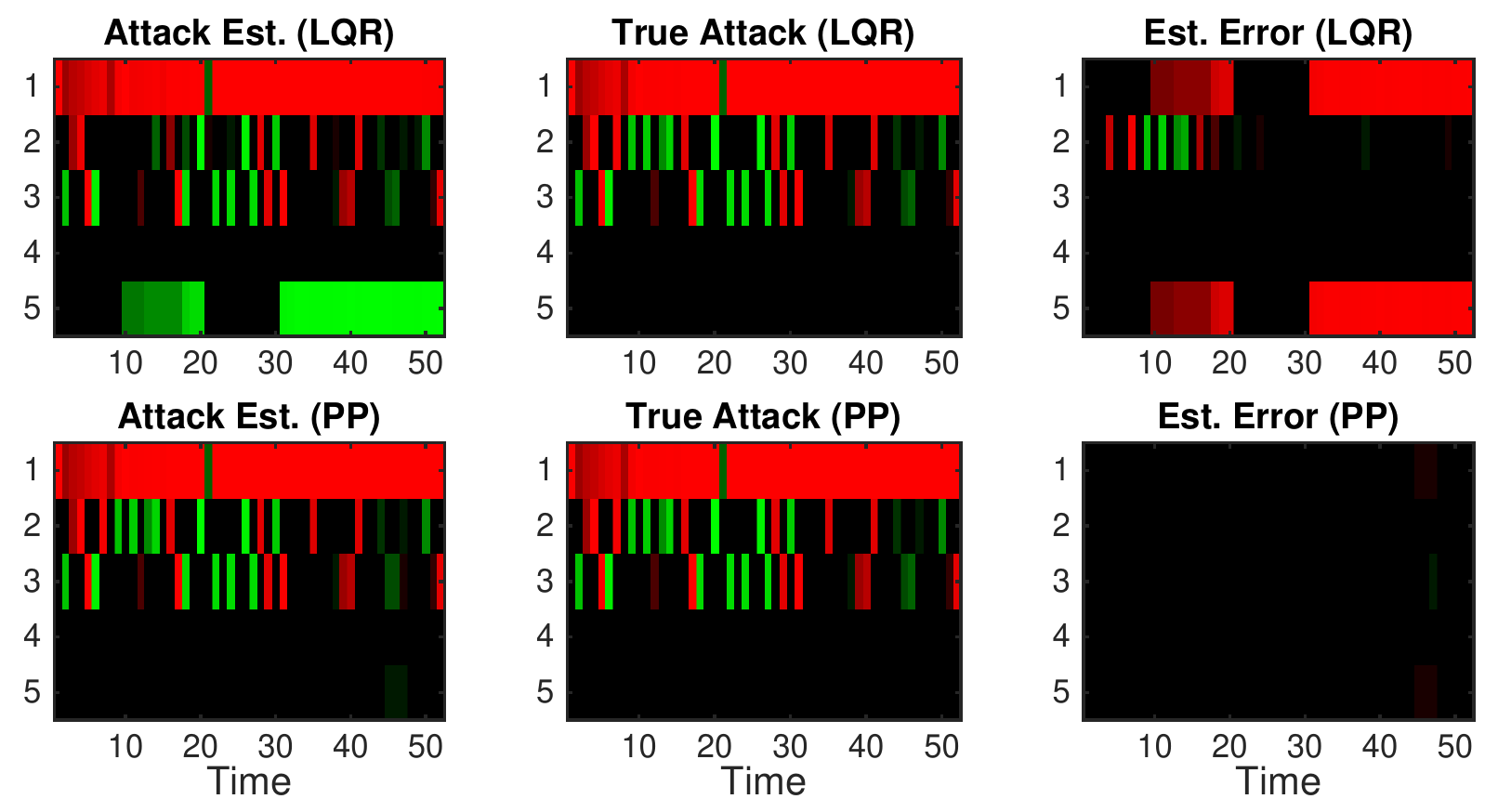}
\caption{Estimated attack signal, true attack signal and estimation error in the attack signal of the estimator (SE) with 2 different feedback controllers: LQR, controller designed via pole-placement (PP); with 5 measurements. Left column shows estimated attack signals. Middle column shows true attack signal. Right column shows estimation error. Each row corresponds to one type of measurement. Red pixels indicate positive values, green pixels are negative values and black indicates zero. }
%\qie{Is it necessary to include this figure? It seems repetitive from previous one.} \yh{It is repetitive from the previous one but it may help to see the structure of $E(t)$, i.e., attacked nodes change over time. Also, it would be good to understand Fig. 9}}
\label{fig:ex_pp_err}
\end{figure}
%%%%%%%%%%%%%%%%%%%%%%%%%%%%%%%%%%%%%

%%%%%%%%%%%%%%%%%%%%%%%%%%%%%%%%%%%%%
\begin{figure}
\center
\includegraphics[width=0.5\textwidth]{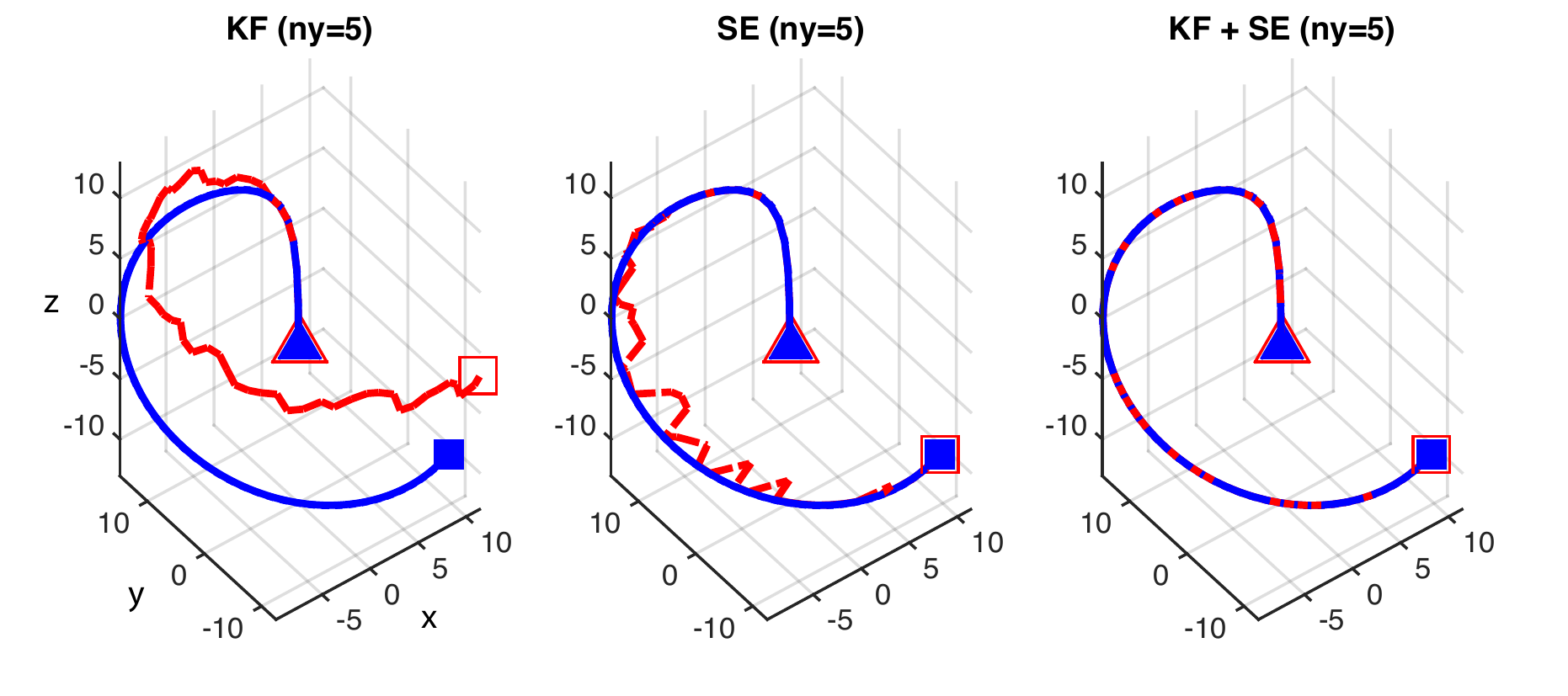}
\caption{Estimated UAV trajectory by three methods under MITM attack: KF only (KF), secure estimator only (SE), secure estimator with KF (KF+SE). Solid blue lines are the true UAV trajectories. They start from the blue triangle and end at the blue square. Red dotted lines represent estimated trajectories by each method, with 5 measurements.}
\label{fig:ex_uav_remote}
\end{figure}
%%%%%%%%%%%%%%%%%%%%%%%%%%%%%%%%%%%%%

\subsubsection{GPS Spoofing}

In this section, we focus on adversarial attacks in the GPS navigation system (Channel 3 in Figure \ref{fig:ex_uav_pic}). Consider the scenario where a UAV uses a Linear Quadratic Gaussian (LQG) controller to follow a desired path, $x_r^{(t)}$, designed by LQ control. In other words, a KF takes compromised and noisy measurements $y^{(t)}$ and outputs a state estimate $\hat x^{(t)}$, which is then used for state feedback control: $u^{(t)} = G (\hat x^{(t)} - x_r^{(t)})$, where $G$ is the feedback matrix. Note that in the previous example (Section \ref{sec:uav_utm}), the feedback controller had access to uncorrupted state measurements $x^{(t)}$, therefore the true path of the UAV is unaffected by attacks. On the other hand, in this example, the UAV uses estimated states $\hat x^{(t)}$ for feedback control and path following. Hence, if measurements are corrupted and the state estimates are poor, then the UAV may not be able to follow its desired path and may deviate away from it. The goal is to correctly estimate the true states of the UAV and therefore, follow the desired path. Assume an attacker spoofs the GPS position measurements in order to deviate the UAV from its planned path. He/she injects a sinusoidal signal to $x$-position measurement, as well as a Gaussian noise to a randomly chosen position measurement at each time step. 

In this example, we explore the effect of the number of sensor measurements on the secure estimation performance of two schemes: (a) KF only, (b) KF+SE.
We first assume that the UAV only uses GPS for navigation, i.e., 3 positional measurements. 
Figure \ref{fig:ex_uav_error} shows that KF completely fails to estimate the attack signal (KF, $n_y = 3$, plots in Row 1), %There are large estimation errors in its estimated $x$- and $z$-positions, 
consequently the actual UAV trajectory (red dashed line)  deviates significantly from its desired path (solid blue line) as shown in Figure \ref{fig:ex_uav_traj}, and deviations are largest along the $x$- and $z$-axis (Figure \ref{fig:ex_uav_est}).
On the other hand, Figures \ref{fig:ex_uav_error} (KF + SE, $n_y = 3$) shows that KF+SE's estimated attack signals are significantly more accurate with only a small estimation error in the $x$-position (plots in Row 2). 
Therefore the UAV can follow its planned path much more closely (Figures \ref{fig:ex_uav_traj}  and \ref{fig:ex_uav_est}).
Recall from Proposition \ref{prop:maximum} that the maximum number of correctable errors for a system with $p$ measurements is $\lceil p/2-1 \rceil$, which equals 1 in this case. There are at most 2 attacked nodes at any time $t$ in this example, which exceeds the above limit. This explains the estimation error in the $x$-position. Despite this small estimation error, the combined scheme KF+SE still outperforms the KF on its own.

We now show the effect of increasing the number of measurements ($n_y$, or equivalently $p$) through sensor fusion, on the estimation performance and consequently, the UAV's path following performance. Autonomous UAVs often use IMUs in addition to GPS for navigation, the former provides additional measurements such as the UAV's velocities, pitch and roll angles. Figure \ref{fig:ex_uav_error} shows that increasing the number of measurements has no affect on the KF's estimation accuracy (compare plots in Rows 1, 3 and 5). 
Even when 8 measurements are used the UAV equipped with a KF still fails to follow the desired path (Figures \ref{fig:ex_uav_traj} and \ref{fig:ex_uav_est}). On the other hand, increasing the number of measurements improves the estimation performance of the secure decoder SE and consequently the performance of the combined scheme KF+SE (compare plots in Rows 2, 4 and 6 in Figure  \ref{fig:ex_uav_error}). Observe that for both 3 and 5 measurements, the combined scheme KF+SE perfectly estimate the attack signals and therefore, can completely subtract them out from the corrupted measurements. As a result, the UAV can follow its original planned path perfectly (KF + SE $n_y=5$, KF + SE $n_y=8$ in Figures \ref{fig:ex_uav_traj} and \ref{fig:ex_uav_est}).

%%%%%%%%%%%%%%%%%%%%%%%%%%%%%%%%%%%%%
\begin{figure}
\center
\includegraphics[width=0.45\textwidth]{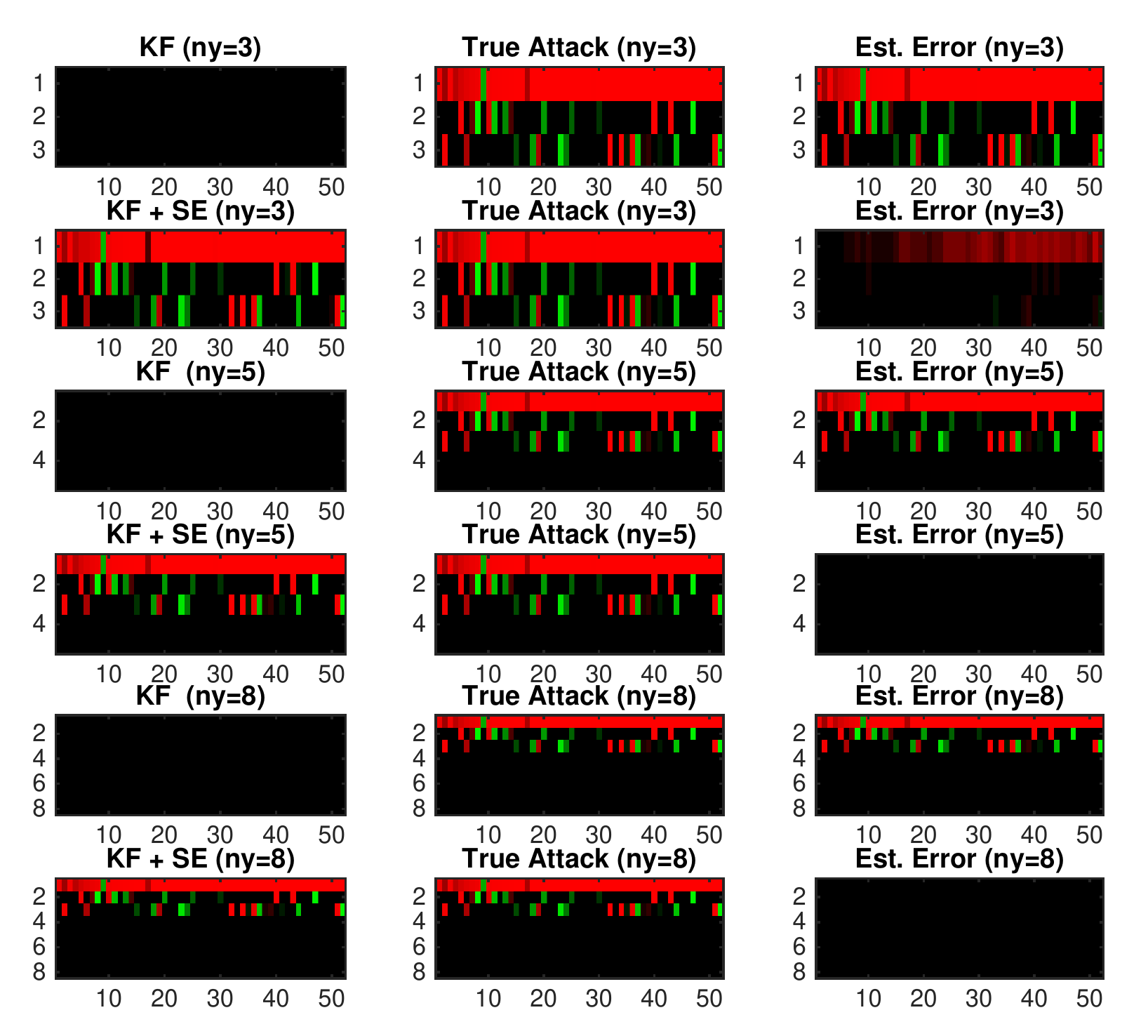}
\caption{Estimated attack signal, true attack signal and estimation error in different cases: KF and KF+SE, each using 3, 5 and 8 different measurements. Left column shows estimated attack signals. Middle column shows true attack signal. Right column shows estimation error. Each row corresponds to one sensor measurement and the first three rows in each plot are the $x$, $y$ and $z$ position measurements, respectively. Red pixels indicate positive values, green pixels are negative values and black indicates zero.}
\label{fig:ex_uav_error}
\end{figure}
%%%%%%%%%%%%%%%%%%%%%%%%%%%%%%%%%%%%%

%%%%%%%%%%%%%%%%%%%%%%%%%%%%%%%%%%%%%
\begin{figure}
\center
\includegraphics[width=0.5\textwidth]{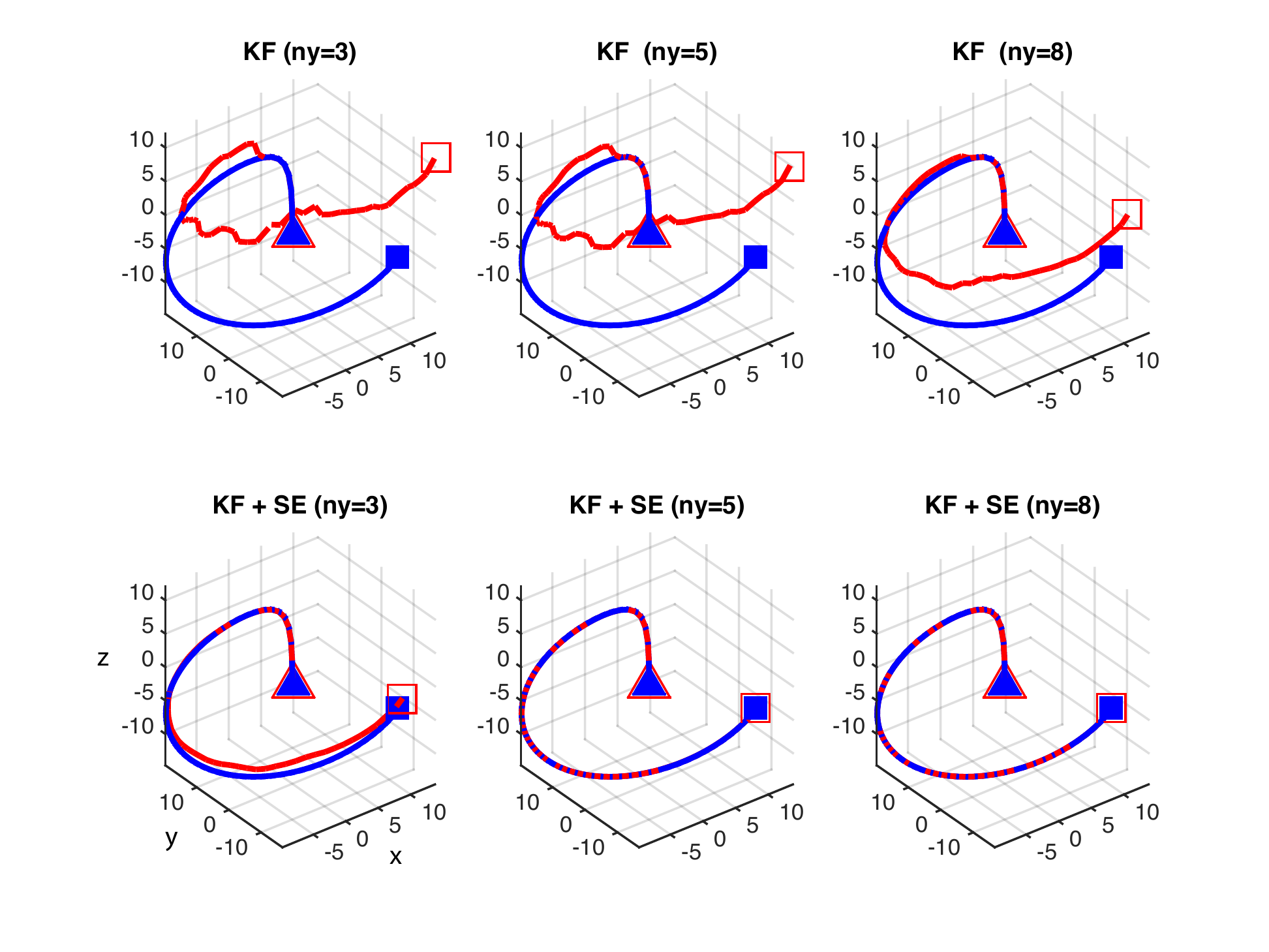}
\caption{Desired and actual UAV trajectory in different cases: KF and KF+SE, each using 3, 5 and 8 different measurements. Blue solid lines are the desired trajectory. Red dash lines are the actual UAV trajectory under adversarial attack.}
\label{fig:ex_uav_traj}
\end{figure}
%%%%%%%%%%%%%%%%%%%%%%%%%%%%%%%%%%%%%

%%%%%%%%%%%%%%%%%%%%%%%%%%%%%%%%%%%%%
\begin{figure}
\center
\includegraphics[width=0.5\textwidth]{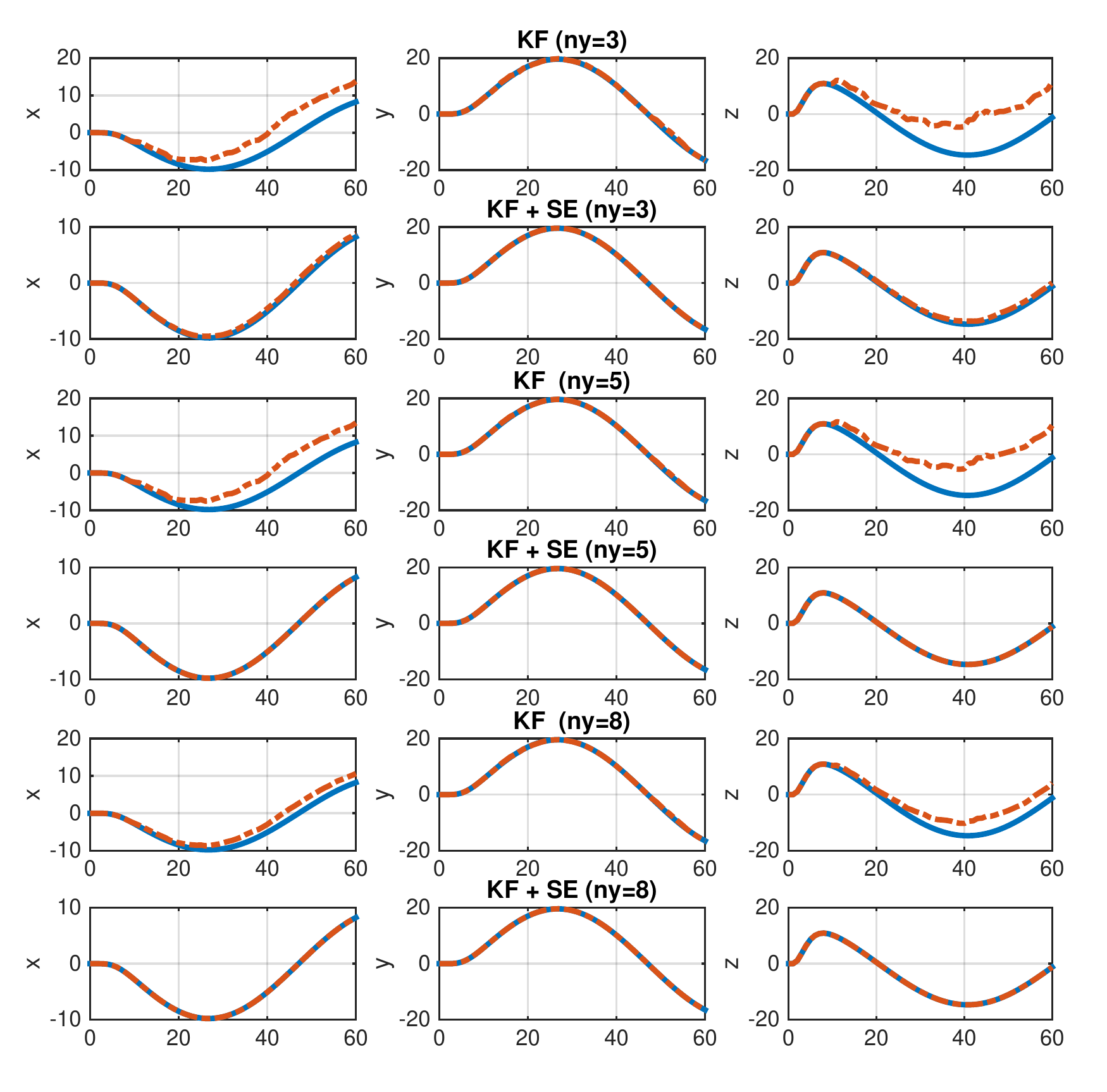}
\caption{Desired and actual UAV trajectory in different cases: KF and KF+SE, each using 3, 5 and 8 different measurements. Blue solid lines are the desired $x$, $y$ and $z$ trajectories. Red dashed lines are the actual UAV trajectories under adversarial attack.}
%\qie{Same for this figure} \yh{What does it mean?}}
\label{fig:ex_uav_est}
\end{figure}
%%%%%%%%%%%%%%%%%%%%%%%%%%%%%%%%%%%%%

\section{Conclusion}
In this paper, we consider the problem of secure estimation for CPS under adversarial attacks. Unlike \cite{Fawzi2014} where the attacked nodes are assumed to be fixed, we allow the set of attacked nodes to change over time, and propose a computationally efficient secure decoder for the latter scenario that works for arbitrary and unbounded attacks. In addition, we propose to combine the secure decoder with a KF for improved practical performance. We demonstrate through numerical examples, that our proposed secure estimator based KF outperforms standard KF. Furthermore, we illustrate practical applications of secure estimation in UAVs under adversarial cyber attacks. This is important not only for today's aviation system but also UAV delivery systems in the near future.

% !TEX root = IEEE_adversarial_attacks.tex
\section*{APPENDIX}

\subsection*{(Proof of Theorem 1)}
In the following lemmas and proposition, we assume the following:
\begin{enumerate}
\item $A \in \mathbb{R}^{n\times n}$ has $n$ distinct positive eigenvalues such that $0 < \lambda_1 < \lambda_2 < \cdots < \lambda_n$.% \st{and each row of $C$ is not identically zero and is not redundant}.
\item $C \in \mathbb{R}^{p \times n}$ is full rank. 
\item The pair $(A,C)$ is observable.
\item $\forall v_i \in \mathbb{R}^n$ where  $ Av_i = \lambda_i v_i $ (i.e., $v_i$ is an eigenvector of $A$), $ \lvert \textsf{supp}(Cv_i) \rvert > 2q$.
\end{enumerate}

Recall that we want to show that given $\lvert \textsf{supp} (\Phi v_i) \rvert > 2q\cdot T $ for all eigenvectors $v_i$ of $A$, then $\lvert \textsf{supp} (\Phi z) \rvert > 2q\cdot T $ for all $z \in \mathbb{R}^n \backslash \{0\}$. Now, $A$ has $n$ distinct eigenvalues, hence the eigenvectors of $A$ form a basis for $\mathbb{R}^n$, and any $z \in \mathbb{R}^n$ can be expressed in the eigenbasis of $A$, i.e., $ z = \sum_{i=1}^n \alpha_i v_i$. Therefore, $\Phi z = \sum_{i=1}^n \alpha_i \Phi v_i$, and thus, the only way that the number of nonzero terms in $\Phi z$ may be less than $2q\cdot T$ is if too many nonzero terms in $\Phi v_i$ are cancelled during this summation. 
In other words, there are rows that are nonzero in the vectors $\Phi v_i$, however after the scaling by $\alpha_i$'s and summation, these rows become zero in $\Phi z$. Therefore our goal is to prove upper bounds on the number of such cancellations, and use these upper bounds to derive a value of $T$ such that even in the worst case (i.e., with most number of cancellations), the number of nonzero terms in $\Phi z$ is greater than $2q\cdot T$ for all $z \in \mathbb{R}^n \backslash \{0\}$.

Before we present the proof, we define the following notations: $c_i^\top$ is the $i$-th row of $C$, ${\bf v} \triangleq \begin{bmatrix} v_1 & v_2 & \cdots & v_n \end{bmatrix} \in \mathbb{R}^{n \times n}$, $\Lambda \triangleq \text{diag}\{ \lambda_1, \cdots , \lambda_n \} \in \mathbb{R}^{n \times n}$. %and $\psi_i \triangleq \Phi v_i = \begin{bmatrix} Cv_i ; \lambda_i C v_i; \lambda_i^2 C v_i; \cdots; \lambda_i^{T-1} C v_i \end{bmatrix}$. 
For all $z \in \mathbb{R}^n$, $ z = {\bf v} \cdot \alpha$ where $\alpha = \begin{bmatrix} \alpha_1, \alpha_2, \cdots, \alpha_n \end{bmatrix}^\top $.
In addition, with these notations, $\Phi z = \begin{bmatrix} C {\bf v} \alpha; C {\bf v} \Lambda \alpha; \cdots; C {\bf v} \Lambda^{T-1} \alpha \end{bmatrix}$.

%We introduce some notions: if there exists $i \in \{1, 2, \cdots, p\}$ such that $c_i^\top v_1 \neq 0$ and $c_i^\top v_2 \neq 0$ but the $i$-th row of $C {\bf v} \Lambda^k \alpha = 0$, then we refer to this as ``there is a cancellation of the $i$-th row at time $k$''. If this is true, then the $i$-th row of $C{\bf v} \Lambda^l \alpha \neq 0$ for all $l \in \{ 0, \cdots, T-1 \}$ where $l \neq k$, i.e., there cannot be another cancellation of the $i$-th row at any other time step $l$, or equivalently, there is a maximum of one cancellation of the $i$-th row over $T$ time steps.

We will first consider two simple cases where $z$ is spanned by two and three eigenvectors ($m = 2$ and $m =3 $) respectively, and then generalize the results to $3 < m \leq n$. The proofs use the following result.

\begin{lem} \label{lem:gv}
For $m$ real numbers $0< \lambda_1 < \lambda_2 < \cdots < \lambda_m$, and $m$ positive integers $0< x_1 < x_2 < \cdots < x_m$, the Generalized Vandermonde matrix $GV(\lambda_1, \cdots, \lambda_m; x_1, \cdots x_m)$ defined as
\begin{equation}
GV(\lambda_1, \cdots, \lambda_m; x_1, \cdots x_m) = 
	\begin{bmatrix} \lambda_1^{x_1} & \lambda_2^{x_1} & \cdots & \lambda_m^{x_1} \\
			\lambda_1^{x_2} & \lambda_2^{x_2} & \cdots & \lambda_m^{x_2} \\
			\vdots & \vdots  & & \vdots\\
			\lambda_1^{x_m} & \lambda_2^{x_m} & \cdots & \lambda_m^{x_m} \\
	\end{bmatrix},
	\end{equation}\nonumber
is nonsingular.
\end{lem}
 
\begin{proof}
$GV(\lambda_1, \cdots, \lambda_m; x_1, \cdots x_m)$ is a submatrix of a Vandermonde matrix $V(\lambda_1, \cdots, \lambda_{T-1})$ defined as
\begin{equation}
V(\lambda_1, \cdots, \lambda_{T-1}) = 
	\begin{bmatrix}1 & 1 & \cdots & 1 \\
			\lambda_1 & \lambda_2 & \cdots & \lambda_{T-1} \\
			\lambda_1^{2} & \lambda_2^{2} & \cdots & \lambda_{T-1}^{2} \\
			\vdots & \vdots  & & \vdots\\
			\lambda_1^{T-1} & \lambda_2^{T-1} & \cdots & \lambda_{T-1}^{T-1} \\
	\end{bmatrix}.
	\end{equation}\nonumber
where $T$ is a positive integer, $0< \lambda_1 < \lambda_2 < \cdots < \lambda_{T-1}$, $\{x_1, x_2, \cdots, x_m\} \subseteq \{0, 1, \cdots, T-1\}$ and $\{ \lambda_1,\cdots, \lambda_m\} \subseteq \{ \lambda_1,\cdots, \lambda_{T-1}\} $.

$V(\lambda_1, \cdots, \lambda_{T-1})$ is a Totally Positive (TP) matrix \cite{fallat2011tnm}, and by definition of TP matrices, all minors of $V(\lambda_1, \cdots, \lambda_{T-1})$, i.e., the determinant of all submatrices of $V(\lambda_1, \cdots, \lambda_{T-1})$, are positive. Therefore $GV(\lambda_1, \cdots, \lambda_m; x_1, \cdots x_m) $ is nonsingular.
\end{proof}

%\textcolor{blue}{(Maybe we don't need Lemma 6 and Lemma 7, can go directly to this Proposition for the general $m$ eigenvector case.)}

\begin{lem}\label{lem:two_vec}
($m=2$)  Consider $z = \sum_{i=1}^2 \alpha_i v_i $ (i.e.,  $\alpha_1 \neq 0$, $\alpha_2 \neq 0$, $\alpha_j = 0, \forall j\ge 3$):
\begin{enumerate}
\item 
If there exists $i \in \{1, 2, \cdots, p\}$ such that $c_i^\top v_1 \neq 0$ and $c_i^\top v_2 \neq 0$ but the $i$-th row of $C {\bf v} \Lambda^k \alpha = 0$, i.e., there is a cancellation of the $i$-th row at time $k$, then the $i$-th row of $C{\bf v} \Lambda^l \alpha \neq 0$ for all $l \in \{ 0, \cdots, T-1 \}$ where $l \neq k$. In other words, there cannot be another cancellation of the $i$-th row at any other time step, or equivalently, there is a maximum of one cancellation of the $i$-th row over $T$ time steps.
\item 
If we choose $T  >  \frac { \min \{s_1, s_2\}} { \max\{s_1, s_2\} - 2q }$, then $\lvert \textsf{supp} (\Phi z) \rvert > 2q\cdot T$. 
\end{enumerate}
\end{lem}

\begin{proof}
(1) %Assume there exists $l \in \{1, \cdots , T-1\}$ such that both $i$-th row of $\lambda_1 C v_1$ and $\lambda_2 C v_2$ are not zero but the $i$-th row of $C {\bf v} \Lambda^l \alpha$ is zero. In other words, if the support of the $i$-th row of $C{\bf v} \Lambda^l \alpha$ is cancelled out, i.e, $c_i ^\top {\bf v} \Lambda^l \alpha = 0$, then for all $k \in \{1, \cdots , T-1\}$ and $k \neq l$, the $i$-th row of $C {\bf v} \Lambda^k \alpha \neq 0$. 
 (Suppose not) %$ c_i^\top {\bf v} \Lambda^k \alpha=0$ where $C = \begin{bmatrix} c_1^\top \\ c_2^\top \\ \vdots \\ c_p^\top \end{bmatrix}$. 
There exist $l \neq k$ and $l \in \{ 0, \cdots, T-1 \}$ such that $c_i^\top {\bf v} \Lambda^l \alpha = c_i^\top {\bf v} \Lambda^k \alpha= 0$:
\begin{equation}
\begin{aligned}
0 = ~& c_i^\top \begin{bmatrix} v_1 & v_2 \end{bmatrix} \begin{bmatrix} \alpha_1 & 0 \\ 0 & \alpha_2 \end{bmatrix} \begin{bmatrix} \lambda_1^l \\ \lambda_2^l \end{bmatrix} \\
= ~ & c_i^\top \begin{bmatrix} v_1 & v_2 \end{bmatrix} \begin{bmatrix} \alpha_1 & 0 \\ 0 & \alpha_2 \end{bmatrix} \begin{bmatrix} \lambda_1^k \\ \lambda_2^k \end{bmatrix} \nonumber
\end{aligned}
\end{equation}
Without loss of generality, assume $l<k$. We can reformulate above equation as follows:
\begin{equation}
	\begin{bmatrix} \lambda_1^{l} & \lambda_2^l \\ \lambda_1^k & \lambda_2^k \end{bmatrix} \begin{bmatrix} \alpha_1 & 0 \\ 0 & \alpha_2 \end{bmatrix} 
	\begin{bmatrix} v_1^\top \\ v_2^\top \end{bmatrix} c_i = \begin{bmatrix} 0 \\ 0 \end{bmatrix}. \nonumber 
\end{equation}
%where $\gamma = \frac{\lambda_2 } {\lambda_1}$.  Since $A$ has $n$ non-zero and distinct magnitude eigenvalues, $| \gamma | \neq 1$ and 
Now, the first matrix on the left hand side (LHS) is a Generalized Vandermonde matrix $GV(\lambda_1,\lambda_2;l, k)$ with $0<\lambda_1< \lambda_2$ and $0< l<k$, thus by Lemma \ref{lem:gv} it is nonsingular. The second matrix on the LHS is also nonsingular as $\alpha_1 \neq 0$ and $\alpha_2 \neq 0$. Therefore we must have $v_1^\top c_i = v_2^\top c_i = 0$ (contradiction, since $c_i^\top v_1 \neq 0$ and $c_i^\top v_2 \neq 0$ by assumption).

(2) Let $L_1, L_2, L_{12}$ be three disjoint subsets of $\{1,\cdots, p\}$ such that $L_1 = \textsf{supp} (C v_1) \cap \textsf{supp}(C v_2)^c $, $L_2 = \textsf{supp} (Cv_2) \cap \textsf{supp} (C v_1)^c $, and $L_{12} = \textsf{supp} (C v_1) \cap \textsf{supp} (C v_2)$ where the superscript $^c$ represents the set complement. Then, $\textsf{supp} (C v_1)  = L_1 \oplus L_{12}$, $\textsf{supp} (C v_2)  = L_2 \oplus L_{12}$, $ s_1 \triangleq \lvert \textsf{supp} (C v_1) \rvert = \lvert  L_1 \oplus L_{12} \rvert > 2q$,  $s_2 \triangleq \lvert \textsf{supp} ( C v_2) \rvert =  \lvert L_2 \oplus L_{12} \rvert > 2q $ and $s_{12} \triangleq \lvert  L_{12} \rvert \le \min \{s_1, s_2\}$. 
Also, possible cancellations only occur in the subset $L_{12}$ by definition. 
\begin{equation}
\begin{aligned}
	\lvert \textsf{supp} (\Phi z) \rvert &= \lvert \textsf{supp} (\alpha_1 \Phi v_1 + \alpha_2 \Phi v_2) \vert  \\&= T \cdot (s_1 - s_{12}) + T\cdot (s_2 - s_{12}) \\ &~+ \sum_{k=0}^{T-1} (s_{12}- s_{r,12}^k) \\
	&=T \cdot (s_1 + s_2 - s_{12}) - \sum_{k=0}^{T-1} s_{r,12}^k \\&
	\ge T  \cdot(s_1 + s_2 - \min \{ s_1, s_2 \}) - \sum_{k=0}^{T-1} s_{r,12}^k \nonumber
\end{aligned}
\end{equation}
where $s_{r,12}^k$ is the number of cancelled support in $L_{12}$ at time step $k$. 
More specifically, $i \in s_{r,12}^k$ if $c_i^\top v_1 \neq 0$ and $c_i^\top v_2 \neq 0$, but $c_i^\top{\bf v}\Lambda^k \alpha = 0$.
From (1), we have the followings:
\begin{equation}
\begin{aligned}
	s_{r,12}^0 &\le \lvert L_{12} \rvert \\
	s_{r,12}^1 &\le \lvert L_{12} \rvert - s_{r,12}^0 \\
	s_{r,12}^2 &\le \lvert L_{12} \rvert - s_{r,12}^0 - s_{r,12}^1 \\
	\vdots & \\
	s_{r,12}^{T-1} & \le \lvert L_{12} \rvert - \sum_{k=0}^{T-2} s_{r,12}^k \nonumber 
\end{aligned}
\end{equation}
Thus, $\sum_{k=0}^{T-1}  s_{r,12}^k \le \lvert L_{12} \vert \le \min \{s_1, s_2 \} $, and
\begin{equation}
\begin{aligned}
	\lvert \textsf{supp} (\Phi z) \vert \ge T \cdot \max\{ s_1, s_2 \} - \min \{ s_1, s_2 \}  > 2q \cdot T .\nonumber 
\end{aligned}
\end{equation}
%Note that $T > \frac { \min \{ s_1, s_2 \} } { \max \{ s_1, s_2 \} - 2q }$.
\end{proof}

\begin{lem} \label{lem:three_vec}$(m=3)$ Consider $z = \sum_{i=1}^3 \alpha_i v_i $ where $\alpha_1 \neq 0$, $\alpha_2 \neq 0$, $\alpha_3 \neq 0$ and $\alpha_i = 0$ for $i = \{4, 5, \cdots , n\}$.   
\begin{enumerate}
\item 
$\sum_{k=0}^{T-1} s_{r,123}^k \le 2 \cdot s_{123}$ where $s_{123} = \lvert L_{123} \rvert = \lvert \textsf{supp} (Cv_1) \cap \textsf{supp} (Cv_2) \cap \textsf{supp} (Cv_3)  \rvert$.
\item  
If we choose $T > \frac { p + \min \{ s_1, s_2, s_3 \}} { \max \{s_1, s_2, s_3 \} - 2q }$,  then $\lvert \textsf{supp} (\Phi z)\rvert > 2q \cdot T$.
\end{enumerate}
\end{lem}
\begin{proof} (1) Claim: the $i$-th row of $C{\bf v}\Lambda^k \alpha$ is cancelled at most 2 times over $T$ time steps, i.e. for at most 2 distinct values of $k \in \{0, 1, \cdots, T-1\}$. 
	
	(Suppose not) There exist 3 distinct time steps $d, e, f \in \{0, 1, \cdots, T-1\}$ such that $c_i^\top {\bf v} \Lambda^d \alpha = c_i^\top {\bf v} \Lambda^e \alpha = c_i^\top {\bf v} \Lambda^f \alpha = 0$. Without loss of generality, assume $d<e<f$:
	\begin{equation}
		\begin{bmatrix} \lambda_1^d & \lambda_2^d & \lambda_3^d \\
					\lambda_1^e & \lambda_2^e & \lambda_3^e \\
					\lambda_1^f & \lambda_2^f & \lambda_3^f 
		\end{bmatrix}
		\begin{bmatrix} \alpha_1 & 0 & 0 \\ 0 & \alpha_2 & 0 \\ 0 & 0 & \alpha_3 \end{bmatrix} {\bf v}^\top c_i 
		= \begin{bmatrix} 0 \\ 0 \\ 0 \end{bmatrix} \nonumber
	\end{equation}

Again, the first matrix on the LHS is a Generalized Vandermonde matrix satisfying the conditions of Lemma \ref{lem:gv}, hence it is nonsingular. The second matrix on the LHS is also nonsingular as $\alpha_1 \neq 0$, $\alpha_2 \neq 0$ and $\alpha_3 \neq 0$. Therefore we must have ${\bf v} ^\top c_i =0 $ (contradiction). A similar derivation as in Lemma \ref{lem:two_vec} then shows $\sum_{k=0}^{T-1} s_{r,123}^k \le 2 \cdot s_{123}$. 

(2) Consider $Cv_1, Cv_2, Cv_3$ and $\Phi z$:\\

\begin{tabular}[!b]{ccc} 
  \hline
  $\textsf{supp}(Cv_1)$ &   $\textsf{supp}(Cv_2)$ &   $\textsf{supp}(Cv_3)$\\
  \hline
	$ L_1$ & {\bf 0 } & {\bf 0} \\
  {\bf 0}  &  $L_2$ & {\bf 0} \\
  {\bf 0} & {\bf 0} &  $L_3$ \\
 $ L_{12}$ &  $ L_{12}$ & {\bf 0 } \\
  {\bf 0} &  $ L_{23}$ &    $ L_{23}$ \\
 $ L_{13}$ & {\bf 0} &   $ L_{13}$  \\
$L_{123}$ &   $L_{123}$ &  $L_{123}$\\
  \hline
\end{tabular}%\\
% {\centering
%\begin{tabular}[!b]{ccc} 
%  \hline
%  $\textsf{supp}(Cv_1)$ &   $\textsf{supp}(Cv_2)$ &   $\textsf{supp}(Cv_3)$\\
%  \hline
%  \cellcolor{red!25}$ L_1$ & {\bf 0 } & {\bf 0} \\
%  {\bf 0}  & \cellcolor{blue!25} $L_2$ & {\bf 0} \\
%  {\bf 0} & {\bf 0} & \cellcolor{green!25} $L_3$ \\
%  \cellcolor{magenta!25} $ L_{12}$ &   \cellcolor{magenta!25} $ L_{12}$ & {\bf 0 } \\
%  {\bf 0} & \cellcolor{cyan!25} $ L_{23}$ &   \cellcolor{cyan!25} $ L_{23}$ \\
%   \cellcolor{yellow!25} $ L_{13}$ & {\bf 0} &   \cellcolor{yellow!25} $ L_{13}$  \\
%  \cellcolor{gray!25}$L_{123}$ &   \cellcolor{gray!25}$L_{123}$ &  \cellcolor{gray!25}$L_{123}$\\
%  \hline
%\end{tabular}\\}
\bigskip\\
\noindent Without loss of generality, assume $s_3\ge s_2 \ge s_1$ (recall $s_1 = \vert L_1 \oplus L_{12} \oplus L_{13} \oplus L_{123} \vert = \lvert L_1 \rvert + s_{12} + s_{13} + s_{123}$):
\begin{equation}
\begin{aligned}
	\lvert \textsf{supp}(\Phi z) \rvert  &=  T \cdot(s_1 - s_{12} - s_{13} - s_{123}) \\&~ + T \cdot(s_2 - s_{12} - s_{23} - s_{123}) \\&~ + T \cdot(s_3 - s_{23} - s_{13} - s_{123}) \\ 
	& \quad+ \sum_{k=0}^ {T-1} \bigg (  (s_{12} - s_{r,12}^k ) + (s_{23} - s_{r,23}^k) \\ &~~~~~~~+ (s_{13} - s_{r,13}^k) + ( s_{123} - s_{r,123}^k) \bigg) \\
	&= T \cdot (s_3 + s_1 + s_2 - s_{12} - s_{13} - s_{23} - 2 \cdot s_{123}) \\&~~~ - \sum_{k=0}^{T-1} ( s_{r,12}^k + s_{r,23}^k + s_{r,13}^k  + s_{r,123}^k ) \\
	& \ge T \cdot s_3 - p - s_{123} \ge T \cdot s_3 - p - \min \{ s_1, s_2, s_3 \} \\
	& > 2q \cdot T
		\nonumber 
\end{aligned}
\end{equation}
where
$\sum_{k=0}^{T-1} (s_{r,12}^k + s_{r,23}^k + s_{r,13}^k   )$ $\le s_{12} + s_{23} + s_{13} $$ \le p - s_{123}$, $\sum_{k=0}^{T-1} s_{r,123}^k \le 2 \cdot s_{123}$ and $ s_{123} \le \min \{ s_1, s_2, s_3 \}$ and note that $T > \frac { p + \min \{ s_1, s_2, s_3 \}} { \max \{s_1, s_2, s_3 \} - 2q }$.
\end{proof}

\begin{prop} \label{prop:m_vec}
Consider $m$ eigenvector combinations $(n\ge m \ge 2)$. Then, the total number of cancellations over $T$ time steps satisfies $\sum_{k=0}^{T-1} s_{r,12 \cdots m}^k \le (m-1)\cdot s_{12 \cdots m}$ where $s_{12\cdots m} = \lvert \textsf{supp} (Cv_1) \cap \cdots \cap \textsf{supp} (Cv_m) \rvert$.
\end{prop}
\begin{proof}
Claim: the $i$-th row of $C{\bf v}\Lambda^k \alpha$ is cancelled at most ($m-1$) times over $T$ time steps, i.e. for at most ($m-1$) distinct values of $k \in \{0, 1, \cdots, T-1\}$
(Suppose not)
\begin{equation}
\begin{aligned}
	&	%\underbrace{
		\begin{bmatrix} \lambda_1^d & \lambda_2^d & \cdots & \lambda_m^d \\
					\lambda_1^e & \lambda_2^e & \cdots & \lambda_m^e \\
					\vdots & \vdots & \ddots & \vdots \\
					\lambda_1^r& \lambda_2^r & \cdots & \lambda_m^r 
		\end{bmatrix} %_{ m \times m}
		\begin{bmatrix} \alpha_1 & 0 & \cdots & 0 \\ 0 & \alpha_2 & \cdots & 0 \\ 
					\vdots & \vdots & \ddots & \vdots \\
					0 & 0 & \cdots & \alpha_m \end{bmatrix} 
					{\bf v}^\top c_i 
		= \begin{bmatrix} 0 \\ 0 \\ \vdots \\ 0 \end{bmatrix} \nonumber\\
%	= &	\underbrace{\begin{bmatrix} 1 & 1 & \cdots & 1 \\
%					\lambda_1^{e-d} & \lambda_2^{e-d} & \cdots & \lambda_m^{e-d} \\
%					\vdots & \vdots & \ddots & \vdots \\
%					\lambda_1^{r-d}& \lambda_2^{r-d} & \cdots & \lambda_m^{r-d} 
%		\end{bmatrix} }_{\text{an $m\times m$ alternant matrix}}
%		\\& \cdot \begin{bmatrix} \lambda_1^d \cdot \alpha_1 & 0 & \cdots & 0 \\ 0 & \lambda_2^d \cdot \alpha_2 & \cdots & 0 \\ 
%					\vdots & \vdots & \ddots & \vdots \\
%					0 & 0 & \cdots & \lambda_m^d \cdot \alpha_m \end{bmatrix} 	{\bf v}^\top c_i 
\end{aligned}
\end{equation}
Again, the first matrix on the LHS is a Generalized Vandermonde matrix satisfying the conditions of Lemma \ref{lem:gv}, hence it is nonsingular. The second matrix on the LHS is also nonsingular as $\alpha_i \neq 0$ for all $i \in \{1, 2, \cdots, m\}$. Therefore we must have ${\bf v}^\top c_i = 0$ (contradiction). A similar derivation as in Lemma \ref{lem:two_vec} then shows $\sum_{k=0}^{T-1} s_{r,12\cdots m}^ k \le (m-1) \cdot s_{12\cdots m}$.
\end{proof}

\section*{Acknowledgment}
This work was supported by the NSF CPS project ActionWebs under grant number 0931843, NSF CPS project FORCES under grant number
1239166.



\bibliographystyle{IEEEtran}

\bibliography{reference}

\end{document}